\documentclass[12pt]{article}
\usepackage{amsmath,amsthm,amssymb,euscript,epsf,epsfig,color,bbm}
\usepackage{array,comment}
\usepackage{fancybox}
\usepackage{hyperref} 
\usepackage{graphicx,ytableau}
\usepackage{cite}
\usepackage{hyperref}
\usepackage{xcolor}
\usepackage{young} 
\usepackage{tikz,pgfplots}
\usetikzlibrary{angles,quotes,arrows.meta,chains,matrix,scopes,calc,intersections,positioning,shapes.misc,decorations.markings,shapes.geometric}
\tikzset{cross/.style={cross out, draw=black,thick, minimum size=3*(#1-\pgflinewidth), inner sep=0pt, outer sep=0pt},
cross/.default={3pt}}

\newcommand{\be}{\begin{equation}}
\newcommand{\ee}{\end{equation}}
\newcommand{\bea}{\begin{eqnarray}\displaystyle}
\newcommand{\eea}{\end{eqnarray}}
\newcommand{\nnm}{\nonumber}

\makeatletter
\@addtoreset{equation}{section}
\makeatother

\def\one{{\hbox{ 1\kern-.8mm l}}}
\def\zero{{\hbox{ 0\kern-1.5mm 0}}}

\def\mC{ \mathbb{C}} 
\def\tr{ {\rm{tr}}}

\def\bm{ \mathbf{m} } 
\def\bcC{ \mathbf{\cC} }

\def\cA{{\cal A}} \def\cB{{\cal B}} \def\cC{{\cal C}}
\def\cD{{\cal D}}  
\def\cG{{\cal G}}  \def\cI{{\cal I}}
 \def\cK{{\cal K}} 
\def\cM{{\cal M}} \def\cN{{\cal N}} 
\def\cP{{\cal P}}  
\def\cS{{\cal S}}  
\def\cV{{\cal V}}  
 \def\cZ{{\cal Z}}
\def\trunc{ {\rm trunc} }

\newtheorem{proposition}{Proposition} 
\newtheorem{lemma}{Lemma}

\begin{document}

\hfill{\normalsize QMUL-PH-26-07}

\begin{center} 
{\LARGE \bf   Gauge-string  duality, monomial bases and graph determinants.  
\\
 }

 \medskip

\bigskip

{\bf Garreth Kemp}$^{a,*}$, {\bf Sanjaye Ramgoolam}$^{b,\dag}  $

\bigskip

$^a${\em Department of Physics,}\\
{ \em University of Johannesburg, } \\
{\em  Auckland Park, 2006, South Africa.  }\\
\medskip
$^{b}${\em  Centre for Theoretical Physics, }\\
{\em Department of Physics and Astronomy,} \\
{\em Queen Mary University of London,} \\
{\em  London E1 4NS, United Kingdom. }\\
\medskip
E-mails:  $^{*}$garry@kemp.za.org,
\quad $^{\dag}$s.ramgoolam@qmul.ac.uk

\end{center}

\begin{abstract} 

Questions  at the intersection of the  AdS/CFT correspondence   and quantum information theory motivate the study of projectors in sequences of subalgebras of finite-dimensional commutative associative semisimple algebras $\cA$, obtained by incrementally adjoining one generator at each step to produce a non-linear generating set for $\cA$. 

We define degeneracy graphs, which are finite layered tree graphs whose nodes represent projectors in the successive subalgebras. Using combinatorial properties of the degeneracy graph, we give a simple formula for constructing a linear basis of $\cA$ in terms of monomials in the generators.The nodes can be labelled by formal variables corresponding to the eigenvalues of the generators added at each layer. 

We prove that the construction is compatible with the required counting of projectors in $\cA$, and give explicit constructions of the projectors in terms of the monomials,  in the cases of one- and two-layer degeneracy graphs with arbitrary numbers of nodes. More generally, we  provide  extensive computational evidence for the invertibility of the matrix relating the proposed monomial basis to the projector basis, by evaluating its determinant. In the  1-layer case, this is a Vandermonde determinant. A simple formula  for the non-vanishing determinant in the general layer case is conjectured and supported by the computational data.

The construction is illustrated with examples including centres of symmetric group algebras and maximally commuting subalgebras generated by Jucys–Murphy elements. We outline applications of the monomial basis to algorithms for constructing matrix units in non-commutative semisimple algebras, with relevance to orthogonal bases of multi-matrix gauge-invariant operators and to quantum information theory.

\end{abstract} 

\newpage 

\tableofcontents

\newpage 

\section{ Introduction } 

An interesting property of the character tables of the symmetric groups $S_n$, of all permutations of $\{ 1 , 2, \cdots , n \}$,  is that the list of characters of a small number of conjugacy classes suffices to distinguish all the irreducible representations. 
This is related to a structural property of the group algebra $ \mC ( S_n )$, which can be viewed as the vector space of formal sums of group elements with complex coefficients, and with product defined using the group multiplication. 

\vskip.2cm 

The centre, $ \cZ ( \mC( S_n) ) $ of $ \mC ( S_n)$, is the sub-space which commutes with all $ \mC ( S_n)$. Its dimension is equal to the number of partitions of $n$. It has a basis of conjugacy class sums, labelled by partitions $p $ of $n$, which consist of 
sums over all group elements with the cycle structure determined by $p$. It has another basis labelled by irreducible representations $R$, corresponding to Young diagrams with $n$ boxes, consisting of projectors $P_R$. The coefficients for the change of basis are given in terms of the irreducible  characters 
\bea 
\chi^R_p = \tr (  D^R ( \sigma_p )  ) 
\eea 
where $ D^R ( \sigma_p) $ is the matrix representing a group element $ \sigma_p \in S_n $ with conjugacy class $p$ in the irreducible representation $R$.  

\vskip.2cm 

Let $d_R$ be the dimension of the irrep $R$, equivalently the character for the trivial group element. Consider $ p $ of the form $ [ k , 1^{ n-k}]$, i.e. partitions with one part of length $ k \ge 2 $ and remaining parts of length $1$. This specifies a cycle structure of permutations in $ S_n$, which have one non-trivial cycle of length $k$ and remaining cycles of length $1$. 
Let $ T_k$ be the sum of permutations in the conjugacy class $ [ k , 1^{ n-k}]$ so that 
\bea 
{ \chi^R ( T_k ) \over d_R } = {  |T_k| \chi^R_{ [ k,1^{ n-k}] }  \over d_R } 
\eea
where $ |T_k| $ is the number of $S_n$ group elements in the conjugacy class $ [ k , 1^{ n-k} ] $. 

\vskip.2cm

It turns out that for symmetric groups $ S_2 , S_3 , S_4 , S_5 , S_7 $, the normalised characters 
 ${ \chi^R ( T_2 ) \over d_R } $ uniquely characterise the  irreps $R$, i.e. no two irreps $R$ have the same normalised character for the conjugacy class $ [ 2,1^{ n-2} ]$. The lists of length $2$ consisting of 
 \bea 
 \{ { \chi^R ( T_2 ) \over d_R }  , { \chi^R ( T_3 ) \over d_R }  \} 
 \eea
distinguish all irreps $ R $ of $S_n$  for $ n $ up to $14$. 
The lists of normalised characters
\[
\left\{ \frac{\chi^R(T_2)}{d_R}, \dots, \frac{\chi^R(T_6)}{d_R} \right\}
\]
distinguish all irreps $R$ for $n$ up to $81$ \cite{KempRam}.

\vskip.2cm 

It   was also explained in \cite{KempRam} that this irrep-distinguishing property of  subsets of conjugacy classes is related to the fact that any projector $ P_R \in \cZ ( \mC ( S_n) ) $ can be expressed as a linear combination of a finite number of powers of the class sums for  the conjugacy classes. For example, for any $ n \in \{ 6,8, \cdots , 14 \}$, we can write any projector 
$P_R \in \cZ ( \mC ( S_n ) ) $ as a finite sum 
\bea 
P_{ R } = \sum_{ a , b } c^R_{ a, b } T_2^a T_3^b 
\eea
for some constants $ c^R_{ab}$. 
Thus a finite set of monomials in $T_2 , T_3$ form a spanning set for $ \cZ ( \mC ( S_n ) ) $. 
We  describe this by saying that $ \{  T_2  , T_3  \} $ form a  non-linear generating set for 
$ \cZ ( \mC ( S_n ) ) $ for $ n $ up to $14$. Similarly,  $ \{ T_2 , T_3 , \cdots , T_{ 6 } \} $ form a non-linear generating set for $ \cZ ( \mC ( S_n ) ) $ for $ n \in \{42,43, \cdots  , 79,81  \}$. 

\vskip.2cm 

The study of non-linear generating sets of conjugacy classes and their ability to distinguish irreps $R$ was motivated by the physics of the half-BPS sector local operators in $ \cN=4$ super-Yang Mills theory with $U(N)$ gauge group in connection with the AdS/CFT correspondence \cite{malda,gkp,witten}.  The half-BPS sector consists of polynomial holomorphic gauge invariant functions constructed from a complex  matrix $Z$. An orthogonal basis, in the CFT inner product, for polynomials of degrees $n$ can be labelled by Young diagrams $R$ with $n$ boxes and no more than $N$ rows \cite{CJR}. The construction of the basis elements is directly related to the projectors $P_R \in \cZ ( \mC ( S_n ) ) $. The Young diagram operators are related to half-BPS geometries for large $n$ \cite{LLM}. 

\vskip.2cm 

The identification of the half-BPS geometries using asymptotic multipole moments of the gravitational fields \cite{IILoss}, and using one-point functions in the CFT \cite{SkTa}, remain active areas of interest in AdS/CFT and inform ongoing discussions on information loss in black hole physics.   The multipole moments are related to Casimirs of $U(N)$ \cite{IILoss}, which by Schur-Weyl duality, are related to central elements in $\cZ ( \mC ( S_n ) ) $. In this setting the consideration of small CFT probes with increasing classical dimension, equivalently increasing energy in the AdS dual,  is related to the consideration of algebras obtained from  sequences $ \{ T_2 , T_3 , \cdots \} $ obtained by incrementally adding a generator. This led to the investigation of an integer sequence $k_* (n ) $ defined, for each $n$, to be the minimal positive integer such that $ \{ T_2 , T_3 , \cdots , T_{ k_*(n) } \} $ form a non-linear generating set for $ \cZ ( \mC ( S_n ) ) $ and have normalised characters which identify any Young diagram $R$ among the irreps of $S_n$. Lower and upper bounds on the large $n$ growth of $k_*(n)$ were obtained in \cite{ProjDetect} and \cite{KempDominance} respectively. 

\vskip.2cm 

Minimal non-linear generating sets are used in eigenvalue systems for constructing integer vectors in the vector space spanned by ribbon graphs, which are enumerated by Kronecker coefficients \cite{QMRibb}. For a wider perspective on current research on Kronecker coefficients and mathematical applications of their connections to tensor invariants, see \cite{LiZhangXia2025}; for emergent phases dominated by ribbon-graph-like structure in the statistical thermodynamics of random regular graphs, see \cite{GorVal}. 

\vskip.2cm 

The generating sets also  arise in eigenvalue systems used to construct orthogonal bases of multi-matrix invariants \cite{PRS}. The computational complexity of quantum algorithms for discriminating projectors $P_R$ using minimal generating subsets of conjugacy classes of $S_n$ was investigated in \cite{ProjDetect}, with the interesting result that the complexities are polynomial in $n$, despite the fact that the number of Young diagrams grows as $e^{\sqrt{n}}$ at large $n$. Reference \cite{ProjDetect} also considered related projection operators— related to  the Wedderburn-Artin matrix basis for an algebra $\cK(n)$ —connected to Kronecker coefficients, and found similar polynomial scaling. The verification of non-vanishing projectors related to  Kronecker coefficients using a different quantum algorithm, also of polynomial complexity,  was studied in \cite{BCGHZ}.

\vskip.2cm

Beyond gauge-string duality and complexity questions arising therein, a related motivation for studying minimal generating sets of conjugacy classes arises from viewing amplitudes of low-dimensional topological field theories. TQFTs based on finite groups $G$ \cite{DW,FHK} can be viewed as a   constructive framework, complementary to Galois-theoretic approaches,  for analysing integrality, positivity and duality properties that relate representation-theoretic data to group multiplication endowed with geometric structure \cite{IDFCTS,RS,PRSe,CSCCT,STV,RowCol}.  

\vskip.2cm 

This background work focused on the detection of projectors using minimal generating sets of conjugacy classes. A structurally complementary problem concerns explicit construction: given such a generating set, what is the complexity of building the projectors themselves? Motivated by quantum algorithms, two-dimensional topological field theory and AdS/CFT, we are thus led to the following question. Given a non-linearly generating set of conjugacy classes for the centre of a group algebra $ \cZ(\mC(G)) $, is there an algorithm which, taking the character table of $G$ as input, constructs a basis for $ \cZ(\mC(G)) $ expressed as monomials in the elements of the non-linearly generating set? This paper answers this question in the affirmative. This constitutes our first main result.

\vskip.2cm 

An important observation is that the question admits a natural formulation in the broader setting of finite-dimensional commutative associative semisimple (CASS) algebras. Semisimplicity implies that such algebras are equipped with a non-degenerate trace pairing. By the Wedderburn–Artin theorem, any such algebra admits a projector basis generalising that of $\cZ(\mC(G))$.

\vskip.2cm

Consider such an algebra $\cA$ of finite dimension $D$ with a minimal non-linear generating set $\{\cC_1, \cC_2, \cdots, \cC_L\}$, where $\cC_i \in \cA$. We take minimality to mean that no proper subset of these generators forms a non-linear generating set.  We consider an ordered sequence of subalgebras
\bea
\cA_1 \rightarrow \cA_2 \rightarrow \cdots \rightarrow \cA_L \equiv \cA
\eea
of increasing dimensions $D_1 < D_2 < \cdots < D_L \equiv  D$, where $\cA_1$ is generated by $\cC_1$, $\cA_2$ is generated by $\{\cC_1,\cC_2\}$, and so forth. Each algebra in the sequence has a basis of projectors. 

\vskip.2cm 

In Section \ref{sec:degengraphs} we five the combinatorial construction of a layered degeneracy graph with $L$ layers associated to any such sequence of algebras. At layer $i$, the nodes correspond to projectors in $\cA_i$ and carry a label denoting the eigenvalue of $\cC_i$ on the corresponding projector. 

\vskip.2cm

The edges connecting nodes at layer $1$ to those at layer $2$ are determined by a partition $p_1$ of $D_2$ with $D_1$ parts. For $2 \le i \le L-1$, the connectivity from layer $i$ to layer $i+1$ is determined by compositions $c_i$ of $D_{i+1}$ with $D_i$ parts. These are expressions of $D_{i+1}$ as a sum of $D_i$ positive integers, where different orderings are treated as distinct compositions; forgetting the order yields a partition of $D_{i+1}$ with $D_i$ parts.

\vskip.2cm

This layered structure leads naturally to a combinatorial description of candidate monomials forming a linear basis for $ \cA$.
In section \ref{sec:monombasis} we state our main conjecture: that there exists a monomial basis for $\cA_L$ determined by the degeneracy graph and given explicitly in \eqref{mainconj}. This specifies a basis  set of monomials which we refer to as   $ \cS ( \cA_1 \rightarrow \cA_2 \rightarrow \cdots \rightarrow \cA_L) $.  As a first consistency check, we give, in section \ref{sec:Counting},  a general counting proof which shows that the number of monomials is equal to the dimension of $ \cA_L $.  We also give the complete proof of the validity of the conjecture for $ L=1,L=2$ in section \ref{sec:PfConst}. 

\vskip.2cm

The monomial basis conjecture implies that the matrix relating the monomials in the set $ \cS ( \cA_1 \rightarrow \cA_2 \rightarrow \cdots \rightarrow \cA_L) $ to the projector basis of $ \cA_L$ is invertible. In section \ref{sec:GraphsDets} we give a formula for the matrix elements of this change of basis (equation \eqref{MatCoeffs}), in terms of the eigenvalues of $ \cC_i$ and conjecture a general form for the non-vanishing determinant of the matrix \eqref{DetConj2}. 
There is substantial computational evidence for this determinant conjecture. The code written in SAGE is available as an ancillary file with the arXiv submission.  A guide to the code, along with examples of degeneracy graphs,  is given in Appendix \ref{sec:AppCode}. 

\vskip.2cm 

Section \ref{sec:Props} discusses some properties of the conjectured monomial basis and section \ref{sec:Apps} explains three applications.

\vskip.5cm

\section{Degeneracy graphs  and generating sequence of CASS algebras }
\label{sec:degengraphs} 

A degeneracy graph is a layered graph. It may be visualised as a sequence of lines, each populated by a finite number of nodes. Choose a positive integer, $L$. This number specifies the number of layers, or depth, of the graph. A node at layer $i$ connects to a subset of nodes at layer $i+1$. An example  with $L=2$ is 
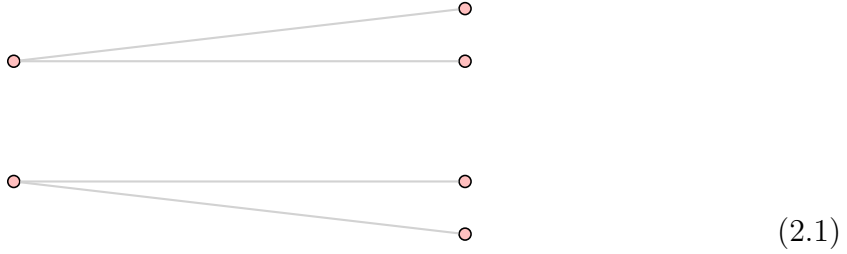
\begin{equation}\label{examp1} 
\begin{tikzpicture}[x=1cm,y=1cm]
	
	\tikzset{
		pt/.style={circle, draw=black, fill=red!25, line width=0.6pt, inner sep=1.6pt},
		ed/.style={draw=gray!35, line width=0.8pt}
	}
	
	% --- coordinates ---
	\node[pt] (x1_2) at (0,  1.0) {};
	\node[pt] (x1_1) at (0, -0.6) {};
	
	\node[pt] (x2_4) at (6,  1.7) {};
	\node[pt] (x2_3) at (6,  1.0) {};
	\node[pt] (x2_2) at (6, -0.6) {};
	\node[pt] (x2_1) at (6, -1.3) {};
	
	% --- edges ---
	\draw[ed] (x1_2) -- (x2_3);
	\draw[ed] (x1_2) -- (x2_4);
	\draw[ed] (x1_1) -- (x2_2);
	\draw[ed] (x1_1) -- (x2_1);
	
\end{tikzpicture}
\end{equation}
and an example with $L=3$ is 
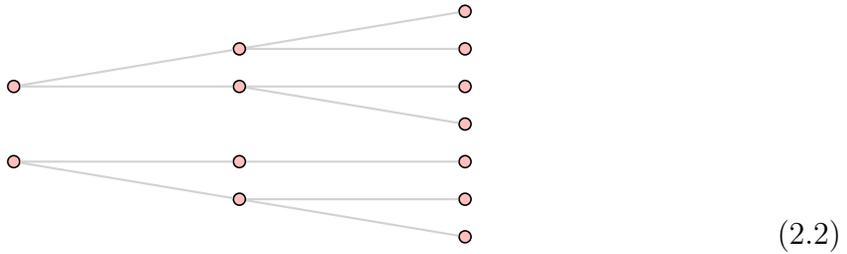
\begin{equation}\label{examp2} 
	\begin{tikzpicture}[x=1cm,y=1cm]
		
		\tikzset{
			pt/.style={circle, draw=black, fill=red!25, line width=0.6pt, inner sep=1.6pt},
			ed/.style={draw=gray!35, line width=0.8pt}
		}
		
		% --- coordinates ---
		% Layer 1
		\node[pt] (x1_2) at (0,  2.2) {};
		\node[pt] (x1_1) at (0,  1.2) {};
		
		% Layer 2
		\node[pt] (x2_4) at (3,  2.7) {};
		\node[pt] (x2_3) at (3,  2.2) {};
		\node[pt] (x2_2) at (3,  1.2) {};
		\node[pt] (x2_1) at (3,  0.7) {};
		
		% Layer 3
		\node[pt] (x3_7) at (6,  3.2) {};
		\node[pt] (x3_6) at (6,  2.7) {};
		\node[pt] (x3_5) at (6,  2.2) {};
		\node[pt] (x3_4) at (6,  1.7) {};
		\node[pt] (x3_3) at (6,  1.2) {};
		\node[pt] (x3_2) at (6,  0.7) {};
		\node[pt] (x3_1) at (6,  0.2) {};
		
		% --- edges layer 1 -> 2 ---
		\draw[ed] (x1_2) -- (x2_4);
		\draw[ed] (x1_2) -- (x2_3);
		\draw[ed] (x1_1) -- (x2_2);
		\draw[ed] (x1_1) -- (x2_1);
		
		% --- edges layer 2 -> 3 ---
		\draw[ed] (x2_4) -- (x3_7);
		\draw[ed] (x2_4) -- (x3_6);
		
		\draw[ed] (x2_3) -- (x3_5);
		\draw[ed] (x2_3) -- (x3_4);
		
		\draw[ed] (x2_2) -- (x3_3);
		
		\draw[ed] (x2_1) -- (x3_2);
		\draw[ed] (x2_1) -- (x3_1);
		
	\end{tikzpicture}
\end{equation}

The more detailed combinatorial characterisation of the graphs will be given below. 
It is motivated by the study of centres of group algebras as explained in the introduction, and the general set-up is  that of  finite dimensional commutative  associative semi-simple (CASS)  algebras. A  CASS algebra  $\cA_L$  over the complex numbers $ \mC $, of dimension $D_L$,  has a basis of projectors $\{ P_I : 1 \le I \le D_L \} $ obeying 
\bea\label{prodProj}  
P_I P_J = \delta_{ IJ } P_{ I } 
\eea
with identity given by 
\bea\label{expI}      
\mathbf{ 1 }  = \sum_{ I } P_I 
\eea
There is a non-degenerate bilinear pairing 
\bea\label{pairing} 
\langle P_I , P_J \rangle = \delta_{ IJ} 
\eea
This algebra can be realised as diagonal matrices $ {\rm Diag }_{ D_L } ( \mC ) $ of size $D_L$ with complex entries. $P_I $ maps to the diagonal matrix with $1$ in the $I$'th entry and zeroes elsewhere. The identity is the  $ \mathbf{ 1 } $ is the unit matrix, and the pairing of two matrices $A,B$ is $ \tr ( A B) $. The Wedderburn-Artin theorem (see for example \cite{CurtRein}) for semi-simple associative algebras, specialised to the commutative case, implies that any CASS algebra has such a basis. Examples of interest in physics include  the centres of group algebras 
as well as maximally commutative sub-algebras of group algebras, such as the algebra generated by Jucys-Murphy elements in the group algebras of symmetric groups. 

We consider the finite dimensional algebra $ \cA_L$, along with a choice of an ordered list of  elements  $ \cC_1 , \cC_2 , \cdots , \cC_L$ which have the property that these elements form a minimal generating set for     the CASS  algebra $ \cA_L $ of finite dimension $D_L$.  Each generator is a finite linear combination of the  projectors.  We have a sequence of  finite-dimensional sub-algebras  
\bea\label{algebras} 
	\cA_1 \rightarrow \cA_2 \rightarrow \cdots \rightarrow \cA_L.
\eea
$\cA_1$ is generated by $ \cC_1 $  : it is the space spanned as a vector space over $ \mC$ by 
$ \{ \mathbf{1}  , \cC_1 , \cC_1^2 , \cdots , \cC_1^{ D_1 -1 } \} $, so that it has dimension $ D_1$. $ \cA_{ 2} $ is generated by $\{  \cC_1 , \cC_2 \} $  : it is spanned by $  \mathbf{1}  $ along with monomials in $ \cC_1 , \cC_2$ with the condition that it has dimension $ D_2$ as a vector space over $ \mC$. For all 
$i$, $ \cA_{ i} $ is generated by $ \{ \cC_1 , \cdots , \cC_i \} $ and has vector space $ D_i$.  All these sub-algebras are semi-simple, with the non-degenerate pairing obtained by specialising \eqref{expI} from $ \cA_L$. They  each have a projector basis by the Wedderburn-Artin theorem.

\subsection{Degeneracy graph data } 
\label{sec:degrapdat} 

The structure of the degeneracy graph is specified by the following data.
\begin{enumerate}
  \item A sequence of positive integers $D_1 , D_2 , \cdots , D_{L} $ obeying the inequalities 
  \bea
  	D_1 < D_2 < \cdots < D_{L-1} < D_L.
  \eea
 The number of nodes in layer $i$ is given by $D_i$.  

  \item A partition $p_{1}$ of $D_2$ with $D_1$ parts. The positive integer parts are organised into weakly increasing order
  \bea
  	p_{1} = \{ p_{1,1} , p_{1,2} , \cdots , p_{1,D_1} \}, \hspace{20pt} p_{1,1} \le p_{1,2} \le  \cdots \le p_{1,D_1}, \hspace{20pt} \sum^{D_1}_{a=1} p_{1,a} = D_2.
  \eea
  The multiplicities of the parts in $p_1$ are defined to be
  \bea
  \label{eq:Multip1}
  	m_j(p_{1}) = \hbox{The number of occurrences of $j$ among the parts of $p_1$.}
  \eea
  Equivalently 
  \bea 
  m_j ( p_1 ) = \sum_{ a=1}^{ D_1} \delta ( j , p_{ 1,a} ) 
  \eea
  
  \item A composition $c_2$ of $D_3$ into $D_2$ parts. We write $c_2$ as an ordered list of positive integers
  \bea
  	c_2 = \{c_{2,1} , c_{2,2} , \cdots , c_{2,D_2}  \}, \hspace{20pt} \sum^{D_2}_{a=1} c_{2,a} = D_3.
  \eea
  \item Compositions $c_i$ of $D_{i+1}$ into $D_i$ parts for $2 \leq i \leq L-1$ also written as an ordered list of positive integers
  \bea
  	c_i = \{c_{i,1} , c_{i,2} , \cdots , c_{i,D_2}  \}, \hspace{20pt} \sum^{D_i}_{a=1} c_{i,a} = D_{i+1}.
  \eea
  Similarly to (\ref{eq:Multip1}), we can define
  \bea
  \label{eq:Multipi}
  	m_j(c_{i}) & =&  \hbox{The number of occurrences of $j$ among the parts of $c_i$.}\cr 
  	           & = & \sum_{ a =1}^{ D_i } \delta ( j , c_{ i , a } ) 
  \eea
  
  \item This data  $ \cD = \{ D_1 , \cdots, D_L ; p_1 , c_2 , \cdots , c_{ L-1} \} $ is used to  specify a  degeneracy graph. 
  
  \item The above data is used to specify a partition of the set $\{ 1, 2 , \cdots , D_2\} \equiv \left[ D_2 \right]$ into $D_1$ successive blocks $B^{(2)}_{a}$ of size $p_{1,a}$,
  \bea
  	\left[ D_2 \right] &=& B^{(2)}_{1} \sqcup B^{(2)}_{2} \cdots B^{(2)}_{D_1} , \nnm\\
	&&\hbox{cardinality of $B^{(2)}_{a}$} = \left| B^{(2)}_{a} \right| = p_{1,a}.
  \eea
  
  \item A partition of the set $ \left[ D_3 \right] \equiv \{ 1, 2 , \cdots , D_3\} $ into $D_2$ successive blocks $B^{(3)}_{a}$ of sizes $c_{2,a}$,
  \bea
  	\left[ D_3 \right] &=& B^{(3)}_{1} \sqcup B^{(3)}_{2} \cdots B^{(3)}_{D_2} , \nnm\\
	&&\hbox{cardinality of $B^{(3)}_{a}$} = \left| B^{(3)}_{a} \right| = c_{2,a}.
  \eea
\item More generally, we will partition the sets $\left[D_i\right]$, for $3 \leq i \leq L$, into successive blocks $B^{(i)}_{a}$ of sizes $c_{i-1,a}$, with $1 \leq a \leq D_{i-1}$,
 \bea\label{BlockDecomp} 
  	\left[ D_i \right] &=& B^{(i)}_{1} \sqcup B^{(i)}_{2} \cdots B^{(i)}_{D_{i-1}} , \nnm\\
	&&\hbox{cardinality of $B^{(i)}_{a}$} = \left| B^{(i)}_{a} \right| = c_{i-1,a}.
  \eea
\end{enumerate}

\subsection{Projectors and eigenvalue labels  } 
\label{sec:ProjEig}
 
By the Wedderburn-Artin decomposition theorem (e.g. \cite{CurtRein}), specialised to the commutative case, each algebra $ \cA_i$ has a basis of projectors $P_{a}^{ (i) }$ for $ a \in \{ 1 , 2, \cdots , D_i \} $ : 
\bea 
P^{(i)}_a P^{(i)}_{ b} = \delta_{ ab } P^{(i)}_a 
\eea
The projectors $ P^{(i)}_a$ are associated with the nodes at the $i$'th layer of the graph. The expansion of these projectors in terms of the projectors of $ \cA_{ i+1} $ is coded by the degeneracy graph. Thus 
\bea\label{Expitoipl1}  
P^{(i)}_a = \sum_{ b \in B^{(i+1)}_a   \subset [ D_{i+1} ]  } P^{(i+1)}_{ b } 
\eea 
where the blocks $B^{(i)}_a$ are determined by the sequence $  ( p_1 , c_2 , \cdots , c_{ L-1} )$ 
as described above in \eqref{BlockDecomp}.

Given a CASS algebra $ \cA_L$,  equipped with a chain of generator sub-algebras $ \cA_1 \rightarrow \cA_2 \cdots \rightarrow \cA_L$, each of the generators $ \cC_i $ is  a linear combination of 
projectors in $ \cA_L$, with coefficients $ \{  x^{(i)}_a :   1 \le a \le D_i \}  $ : 
\bea 
\cC_i = \sum_{ a =1 }^{ D_i }  x^{(i)}_a P^{(i)}_a 
\eea
and the projector relations imply that 
\bea 
\cC_i P^{(i)}_a = x^{(i)}_a P^{(i)}_a 
\eea
Each $P^{(i)}_a $ is a sum over a subset of  irreducible projectors $ P_I \in \cA_L$,  with coefficient  $1$.  For $ \cA_1$, we have 
\bea 
x^{(1)}_{b_1 } \ne x^{(1)}_{ b_2}   ~~~ \hbox{ for } ~~~ b_1 , b_2 \in [ D_1 ] ~~\hbox{and}~~  b_1  \ne b_2  
\eea
 Adding the generator $ \cC_2$ to $ \cC_1$ gives  $ \cA_2$ of dimension $D_2 > D_1$ and  the  distinct projectors in $ \cA_2 $ correspond to the $ D_2$ nodes of the layered graph at level $2$. Combining these facts with \eqref{Expitoipl1} gives the condition 
\bea 
 ~~~  \hbox{ For   all } a \in [ D_1] :  b_1 , b_2 \in B^{(2)}_{ a } \hbox{ and }  b_1 \ne b_2    \implies ~~~  x^{(2)}_{b_1}  \ne x^{(2)}_{b_2}    
\eea
More generally, for any $  2 \le i \le L $ 
\bea\label{nondegleaves}  
 \hbox{ For  all  } a \in [ D_{i-1} ]  : ~~  b_1 , b_2 \in B^{(i)}_{ a } \hbox{ and }  b_1 \ne b_2   \implies x^{(i)}_{b_1}  \ne x^{(i)}_{b_2}  
\eea
These inequalities on the eigenvalues at each layer ensure that adding the successive combinatorial generators $ \cC_i$ to the generating set produces the sequence $ \cA_1 \rightarrow \cdots  \rightarrow \cA_L $ of  CASS algebras of increasing dimension, with increasing refinement of the projectors described by the graph.

{\bf Summary} The integers $D_i$ give the number of distinct eigenvalue lists for the first $i$ generators. The partitions and compositions encode how degeneracies split when a new generator is adjoined. Together this data determines a finite layered tree which we call the degeneracy graph. The nodes at each layer are projectors in $\cA_i$ determined by  the eigenvalue list  of the first $i$ generators.

\subsection{ The  CASS-algebra in the monomial basis from the degeneracy graph } 
In the main conjecture  \eqref{mainconj} of section \ref{sec:monombasis} we give the conjectured form of a basis set 
 monomials in the generators $ \cC_{ 1}, \cdots , \cC_L $ for $ \cA_L $. 
A labelled version of the degeneracy graph carries variables $x^{(i)}_a$ at each node. These are  distinct  eigenvalues of $ \cC_i$ when applied to the projector $ P^{(i)}_a$:
\bea 
\label{eigvaleqforClassSum}
\cC_i P^{(i)}_a = x^{ (i)}_a P^{(i)}_a
\eea
The labelled versions of the graphs in \eqref{examp1} and \eqref{examp2} are in Appendix \ref{sec:AppCode}. 
Any  specified monomial  
\bea 
\cC_1^{ m_1} \cC_2^{m_2} \cdots \cC_L^{ m_L } = \prod_{ i =1}^{ L } \cC_i^{ m_i } 
\eea
can be evaluated on the projectors of $\cA_L$, using the block decompositions of projectors specified by the sequence $ ( p_1 , c_2 , \cdots , c_{L_1} )$. The  equation for the eigenvalues is given below in \eqref{MonomEigs}. It is convenient to define 
\bea 
 \bcC^{ \bm } = \prod_{ i =1}^{ L } \cC_i^{ m_i } 
\eea
Thus 
\bea 
\label{eq:EigenvalEqforMonos}
 && \bcC^{ \bm } P_I  =   \cM_{ I ,  \bm } P_I  \cr 
 &&  \bcC^{ \bm }  = \sum_{I=1}^{D_L}  \cM_{ I ,  \bm } P_I \cr 
 && P_I = \sum_{\bm } \cM^{-1}_{\bm,  I} \bcC^{ \bm } 
\eea
Importantly as $ \bm $ runs over the set of monomials $   \mathbf {  \cS ( \cA_1 \rightarrow \cA_2 \rightarrow \cdots \rightarrow \cA_L  )  }$ specified in \eqref{mainconj} and 
$ I $ runs over the nodes of the degeneracy graph at the final layer, the conjecture states that the matrix is invertible.

The CASS-algebra structure which is manifest in the projector basis of $\cA_L$, as described in \eqref{prodProj} \eqref{expI} \eqref{pairing}, can be expressed in terms of the monomial basis using the change of basis matrix $  \cM_{ I ,  \bm }$. The product is given by 
\bea 
 \bcC^{ \bm^{(1)}  } \cdot  \bcC^{ \bm^{(2)}  }   
  = \sum_{ I =1}^{ D_L } \sum_{ \bm^{(3)} \in  \mathbf { \cS ( \cD )  } } 
     \cM_{ I , \bm^{(1)} } \cM_{ I , \bm^{(2)} } \cM^{-1}_{  \bm^{(3)} . I  }  
     \bcC^{ \bm^{(3)}  }   
\eea
The trace-pairing is 
\bea 
\langle \bcC^{ \bm^{(1)}   } , \bcC^{  \bm^{(2)}   } \rangle  
 = \sum_{ I =1}^{ D_2 } \cM_{ I, \bm^{(1)} }  \cM_{ I, \bm^{(2)}  }
\eea
And the unit is 
\bea 
\mathbf{1} = \sum_{ I  =1 }^{ D_L }  \cM^{-1}_{ \bm , I  } \cC^{ \bm } 
\eea

%%%%%%%%%%%%%%%%%%%%%%%%%%%%%%%%%%%%%%%%%%%%%%%%%%%%%%%%%%%%%%%%%%%%%%%%%%%%%%%%%%%%%%%%%%%%%%%%%%%%%

%%%%%%%%%%%%%%%%%%%%%%%%%%%%%%%%%%%%%%%%%%%%%%%%%%%%%%%%%%%%%%%%%%%%%%%%%%%%%%%%%%%%%%%%%%%%%%%%%%%%%

\section{A  monomial basis for $  \cA_L$ from degeneracy graph }
\label{sec:monombasis} 

In this section, we propose a formula  \eqref{mainconj} for a  monomial basis of $ \cA =\cA_L $, which uses the sequence of sub-algebras $ \cA_1 \rightarrow \cdots \rightarrow \cA_L $ along with the associated degeneracy graph, labelled with eigenvalues of the generators $ \cC_1 , \cC_2 , \cdots ,\cC_L$. The sequence of sub-algebras, their relation  to degeneracy graphs, and the systematic generation of degeneracy graphs was described in section \ref{sec:degengraphs}. An important building block for the basis is the definition of subsets  $ \cS^{(i)}_{\left[ d_{i+1},d_{i+2} , \cdots ,d_L \right]} \subset \cA_i  $. These are nodes in the $i$'th layer of the degeneracy graph which obey a bound on the numbers of links connecting them to nodes in higher layers labelled by  $j$ with  $ ( i+1) \le  j \le L $.  The set $\cS^{(1)}_{\left[ d_2 , d_3 , \cdots , d_L \right]}$ will play a crucial role in the definition  of the monomial basis. 

\subsection{Definition of $\cS^{(i)}_{\left[ d_{i+1}  , d_{i+2}  , \cdots , d_L \right]}$}
\label{sec:defS1}

Given any $I \in \left[ D_L \right]$, the block decomposition of $\left[D_L\right]$ given by the composition  $c_{ L-1}$, specifies a block $B^{(L)}_{a}$, such that $I \in B^{(L)}_{a}$. Recall that block $B^{(L)}_{a}$ has size $c_{L-1,a}$ for some $a \in \left[ D_{L-1} \right]$. Informally, we say that  $a$ is the parent of $I$ and $c_{L-1,a}$ is the number of daughters of $a$ in $\left[ D_L \right]$ : in the degeneracy graph the node  $I \in [D_L]$ is connected to the node $a \in [ D_{ L-1}]$, $a$  is connected to a total of $ c_{ L-1,a}$ nodes in $[D_L]$.   We define
\bea
	\cS^{(L-1)}_{\left[ d_L \right]} = \{ a \in \left[ D_{L-1} \right] \hbox{ such that } c_{L-1,a} \geq d_L \}.
\eea
This definition gives the set of nodes in $\left[ D_{L-1} \right]$ with $d_L$, or more, daughters. Next we define $\cS^{(L-2)}_{\left[ d_{L-1} , d_L \right]} \subset \left[ D_{L-2} \right]$ by
\bea
	\cS^{(L-2)}_{\left[ d_{L-1} , d_L \right]} = \left\{ a \in \left[ D_{L-2} \right] \hbox{ such that } \left| B^{(L-1)}_{a} \cap \cS^{(L-1)}_{\left[d_L\right]}  \right| \geq d_{L-1} \right\}.
\eea
In general, for all $1 \leq i \leq L-2$, we define $\cS^{(i)}_{\left[ d_{i+1} , d_{i+2} ,\cdots , d_L \right]} \subset \left[ D_i \right]$ by
\bea
\label{eq:cSi}
	\cS^{(i)}_{\left[ d_{i+1},d_{i+2} , \cdots ,d_L \right]} = \left\{ a \in \left[ D_{i} \right] \hbox{ such that } \left| B^{(i+1)}_{a} \cap \cS^{(i+1)}_{\left[ d_{i+2} , \cdots , d_L\right]}  \right| \geq d_{i+1} \right\}.
\eea
The special case $\cS^{(1)}_{\left[ d_2 , d_3 , \cdots , d_L \right]} \subset \left[ D_1 \right]$ is used in defining the monomials. To shorten the notation, we will often use $[\vec{d}] = [d_2 , d_3 , \cdots , d_L]$. 

Equivalently, in words, 
\bea
	\cS^{(1)}_{\left[d_2,d_3 , \cdots , d_{L}\right]} &=& \hbox{the set of all vertices in layer $1$ having $d_2$, or more, daughters in layer $2$,}\nonumber\\
						&& \hbox{each of which have $d_3$, or more, daughters in layer $3$,   continuing until }\nonumber\\
						&& \hbox {the layer $(L-1)$ where each of vertices in layer  $(L-1)$}\nonumber\\
						&& \hbox{have $d_L$, or more, daughters in layer $L$.   }
\label{eq:DefSAL}
\eea

%%%%%%%%%%%%%%%%%%%%%%%%%%%%%%%%%%%%%%%%%%%%%%%%%%%%%%%%%%%%%%%%%%%%%%%%%%%%%%%%%%%%%%%%%%%%%%%%%%%%%

\subsection{Monomial basis  conjecture}
\label{sec:mainconj} 

We define $ {\rm Monom } ( d_2 , d_3 , \cdots , d_L ) $ as a  set of monomials 
\bea
\label{eq:MonoBasisforAL}
{\rm \bf Monom } ( d_2 , d_3 , \cdots , d_L ) = 	\{1, \cC_1 , \cC^2_1 , \cdots ,  \cC^{|\cS^{(1)}_{\left[\vec{d}\right]}|-1}_{1} \} \times \cC^{d_2-1}_{2}\cC^{d_3-1}_{3} \cdots \cC^{d_{L}-1}_{L} . 
\eea
where $ \cS^{(1)}_{\left[\vec{d}\right]}  $ is defined above as a special case of  \eqref{eq:cSi} and equivalently described in \eqref{eq:DefSAL}. 

Our main conjecture is the following.\\
\noindent 
{\bf Monomial Basis Conjecture: }A basis of $ \cA_L$ is the disjoint union 
\bea\label{mainconj}
\cS ( \cA_1 \rightarrow \cdots \rightarrow \cA_L ) \equiv  \bigsqcup_{ d_2 , d_3 , \cdots , d_{ L  } }  {\rm \bf  Monom } ( d_2 , d_3 , \cdots , d_L ) 
\eea
where $ d_i \in \{ 1 , \cdots , d_{i}^{\max} \}  $. The  number of  monomials in  (\ref{eq:MonoBasisforAL}) for each $\vec{d}$ is equal to $\left| \cS^{(1)}_{\left[ \vec{d} \right]} \right|$. The total number of monomials in $\cS ( \cA_1 \rightarrow \cdots \rightarrow \cA_L ) $   is therefore
\bea
\label{eq:TotNumofMonincAL}
	\sum^{d^{\max}_{L}}_{d_L = 1}  \cdots \sum^{d^{\max}_{3}}_{d_{3}=1} \sum^{d^{\max}_{2}}_{d_2 = 1} \left| \cS^{(1)}_{\left[ \vec{d} \right]} \right| 
\eea
For the special case $L=1$, with an algebra of dimension $D_1$, equivalently a graph with a single layer with $ D_1$ nodes, the monomial basis is 
\bea
	\{1, \cC_1 , \cC^2_1 , \cdots ,  \cC^{  D_1 - 1  }  \}.
\eea
The projectors can be written as linear combinations of these monomials, as we will recall in section \ref{sec:const1}. In this case, denoting the  the $D_1$ eigenvalues of $ \cC_1 $ as $ \{ x_a : 1 \le a \le D_1  \} $, the change of basis matrix relating the monomials to the projectors is a standard Vandermonde matrix with matrix entries $ x_a^{ i  } $ with  $ 0 \le a \le D_1 -1  $ 

%We refer to the set of monomials in \eqref{mainconj} as $ \cS ( \cA_1 \rightarrow \cA_2 \rightarrow \cdots \rightarrow \cA_L ) $.\\

%\subsection{  Remark  }

\section{ Counting proof for general degeneracy graphs  }
\label{sec:Counting} 

In this section, we prove that the total number of monomials in $ \cS ( \cA_1 \rightarrow \cdots \rightarrow \cA_L ) $ (defined in \eqref{mainconj}) is equal  to $D_L$, which is  the dimension of $\cA_L$ and the number of nodes in the $L$'th layer of the degeneracy graph.   
\begin{proposition}\label{GenCountProp}
	\bea
\label{eq:CountingforDL}
	\sum^{d^{max}_{L}}_{d_L = 1}  \cdots \sum^{d^{max}_{3}}_{d_{3}=1} \sum^{d^{max}_{2}}_{d_2 = 1} \left| \cS^{(1)}_{\left[ \vec{d} \right]} \right| = D_L . 
	\eea
\end{proposition}
Definitions (\ref{eq:Multip1}), (\ref{eq:Multipi}) will be useful in proving Proposition \ref{GenCountProp}. Using these definitions we can write the dimensions of the layers in the graph in terms of the multiplicities. 
\bea
	D_1 =  \sum\limits^{d^{max}_{2}}_{l=1} m_{l}(p_{1}) \\ 
	D_2 = \sum\limits^{d^{max}_{2}}_{l=1} l m_{l}(p_{1})
\eea
and for $i \geq 2$,
\bea
\label{eq:DimAi}
	D_i &=& \sum\limits^{d^{max}_{i+1}}_{l=1} m_{l}(c_i)\\
\label{eq:DimAip1}
	D_{i+1} &=& \sum\limits^{d^{max}_{i+1}}_{l=1} l m_{l}(c_i). 
\eea
We will also make use of definition (\ref{eq:cSi}). 

\subsection{$\cA_{1}$}

This is a special case in the monomial basis conjecture of section \ref{sec:monombasis} where the monomials are directly specified as 
$\{ 1 , \cC_1 , \cdots , \cC_1^{ D_1 - 1 } \}$ and the count is equal to $D_1$, the dimension of $ \cA_1$. 
\begin{comment} 
Beginning with the case of constructing projectors in $\cA_1$, our graph consists of a single layer. When $d_2 = 1$, the quantity $\left| \cS^{(1)}_{\left[ 1 \right]} \right|$, according to conjecture \ref{mainconj} is the number of $\cC_1$ monomials needed to construct all projectors in $\cA_1$. $\left| \cS^{(1)}_{\left[ 1 \right]} \right|$ is also equal to the number of nodes in $\cA_1$ having $1$, or more, daughters in $\cA_2$. But this is simply $D_1$. Thus, the total number of $\cC_1$ monomials used in constructing the projectors in $\cA_1$ is
\bea
	\left| \cS^{(1)}_{\left[ 1 \right]} \right| = D_1.
\eea
\end{comment}

\subsection{$\cA_{1} \rightarrow \cA_2$}

We now prove Proposition \ref{GenCountProp} for the case of $L=2$.
\begin{lemma}\label{lemma1}
For the case of $L = 2$, 
	\bea
	\sum^{d^{max}_{2}}_{d_2 = 1} \left| \cS^{(1)}_{\left[ d_2 \right]} \right| = D_2. 
	\eea
\end{lemma}
\noindent 
\textbf{Proof:}\\
Recall that $\cS^{(1)}_{\left[d_2\right]}$ is the set of nodes in $\cA_1$ that have $d_2$, or more, daughters in $\cA_2$. Thus, we have the identity 
\bea
\label{eq:IdenSd2}
	\left| \cS^{(1)}_{\left[d_2\right]} \right| = \sum^{d^{max}_{2}}_{l = d_2} m_{l}(p_1).
\eea
It follows that 
\begin{comment}
The LHS is also equal to the number of new monomials introduced at stage $d_2$. Thus, summing $\left| \cS^{(1)}_{\left[d_2\right]} \right|$ over $d_2$ from 1 to $d^{max}_{2}$ will give the total number of monomials used to construct the projectors in $\cA_2$; %From (\ref{eq:IdenSd2}), 
\end{comment}
\bea
	\sum^{d^{max}_{2}}_{d_2=1}\left| \cS^{(1)}_{\left[d_2\right]} \right| &=& \sum^{d^{max}_{2}}_{d_2=1} \sum^{d^{max}_{2}}_{l = d_2} m_{l}(p_1).
\eea
Reversing the order of the summations on the RHS,
\bea
\sum^{d^{max}_{2}}_{d_2=1}\left| \cS^{(1)}_{\left[d_2\right]} \right| &=&  \sum^{d^{max}_{2}}_{l=1} \sum^{l}_{d_2 = 1} m_{l}(p_1),\nnm \\
\label{eq:seclineA2counting}
				&=&  \sum^{d^{max}_{2}}_{l=1} l m_{l}(p_1).
\eea
From (\ref{eq:DimAip1}), the RHS of (\ref{eq:seclineA2counting}) is simply $D_2$, and Lemma \ref{lemma1} is proved. Note that we can use Lemma \ref{lemma1} at any layer in the graph. For instance, at layer $\cA_{i-1}$, beginning from the observation analogous to (\ref{eq:IdenSd2})
\bea
	\left| \cS^{(i-1)}_{\left[d_i\right]} \right| = \sum^{d^{max}_{i}}_{l=d_i} m_{l}(c_{i-1}),
\eea
and following the same steps as above we arrive at
\bea
	\sum^{d^{max}_{i}}_{d_i = 1} \left| \cS^{(i-1)}_{\left[d_i\right]} \right|  = D_{i}.
\eea

\subsection{$\cA_{1} \rightarrow \cA_2 \rightarrow \cA_3$}

\begin{lemma}
For the case of $L = 3$, 
	\bea\label{lem2}
	\sum^{d^{max}_{3}}_{d_3 = 1} \sum^{d^{max}_{2}}_{d_2 = 1} \left| \cS^{(1)}_{\left[ d_2,d_3 \right]} \right| = D_3.
	\eea
\end{lemma}
\noindent 
\textbf{Proof:}\\
We make use of the identity 
\bea
\label{eq:IdenSd2d3}
	\sum^{d^{max}_{2}}_{l=1} \left| \cS^{(1)}_{\left[ l ,  d_3\right]} \right| = \left| \cS^{(2)}_{\left[ d_3 \right]} \right| .
\eea
%Equation (\ref{eq:IdenSd2d3}) says that summing $\left| \cS^{(1)}_{\left[ l ,  d_3\right]} \right|$ over $l$ from 1 to $d^{max}_{2}$ for a fixed $d_3$ gives the number of nodes in $\cA_2$ having $d_3$, or more, daughters in $\cA_3$.
To prove (\ref{eq:IdenSd2d3}), we can visualize a two-layer truncated graph in which all nodes in $\cA_2$ having links to fewer than $d_3$ daughters in $\cA_3$ are dropped. The links of these dropped nodes to their parent nodes in $\cA_1$ are also dropped. Denote this truncated graph by $\cG_{(2)}(d_3)$. Let $D_{2}(\cG_{(2)}(d_3))$ be the total number of nodes in $\cA_2$ in $\cG_{(2)}(d_3)$. Then, by definition 
\bea
\label{eq:D2truncated}
	D_{2}(\cG_{(2)}(d_3)) = \left| \cS^{(2)}_{\left[d_3\right]} \right|.
\eea
Define $\cS^{(1)}_{\left[ l \right]} \left( \cG_{(2)}(d_3 \right)$ to be the set $ \cS^{(1)}_{\left[ l \right]}$ evaluated on the truncated graph $ \cG_{(2)}(d_3)$. Then from lemma \ref{lemma1}, 
\bea
	\sum^{d^{max}_{2}}_{l=1} \left| \cS^{(1)}_{\left[ l \right]} \left( \cG_{(2)}(d_3 \right) \right| = D_{2}\left( \cG_{(2)}(d_3) \right).
\eea
Using (\ref{eq:D2truncated}), we arrive at
\bea
\label{eq:SumS1trancatedequalcS2d3}
	\sum^{d^{max}_{2}}_{l=1} \left| \cS^{(1)}_{\left[ l \right]} \left( \cG_{(2)}(d_3 \right) \right| =  \left| \cS^{(2)}_{\left[d_3\right]} \right|.
\eea
By construction, the number of nodes in $\cA_1$ having $l$, or more, daughters in $\cA_2$ for $\cG_{(2)}(d_3)$ is equivalent to the number of nodes in $\cA_1$ having $l$, or more daughters in $\cA_2$, each of which have $d_3$, or more, daughters in $\cA_3$ for the original graph:
\bea
\label{eq:byconstructcS1}
	 \left| \cS^{(1)}_{\left[ l \right]} \left( \cG_{(2)}(d_3 \right) \right| =  \left| \cS^{(1)}_{\left[ l , d_3 \right]} \right| .
\eea
Combining (\ref{eq:byconstructcS1}) with (\ref{eq:SumS1trancatedequalcS2d3}), we arrive at (\ref{eq:IdenSd2d3}).\\

Summing over $d_3$ from 1 to $d^{max}_{3}$ in (\ref{eq:IdenSd2d3}) gives
\bea
	\sum^{d^{max}_{3}}_{d_3 = 1}\sum^{d^{max}_{2}}_{l = 1}\left| \cS^{(1)}_{\left[ l , d_3  \right]} \right| = \sum^{d^{max}_{3}}_{d_3 = 1}  \left| \cS^{(2)}_{\left[d_3\right]} \right|
\eea
We can apply lemma \ref{lemma1} to the RHS to obtain $D_3$. This completes the proof of Lemma 2 (eqn. \eqref{lem2}).

\subsection{$\cA_{1} \rightarrow \cA_2  \rightarrow \cdots  \rightarrow \cA_L$}

We now prove proposition \ref{GenCountProp}.

\textbf{Proof}
We make use of the identity 
\bea
\label{eq:IdenSd2d3dL}
	\sum^{d^{max}_{L-1}}_{l_{L-1}=1} \cdots \sum^{d^{max}_{3}}_{l_3=1}\sum^{d^{max}_{2}}_{l_2=1}\left| \cS^{(1)}_{\left[ l_2 ,  l_3 , \cdots , l_{L-1}, d_L\right]} \right| = \left| \cS^{(L-1)}_{\left[ d_L \right]} \right| .
\eea
To prove (\ref{eq:IdenSd2d3dL}), we define the following two-layer truncated graph $\cG_{(2)}(d_3 , \cdots , d_L)$. In $\cG_{(2)}(d_3 , \cdots , d_L)$ all nodes in $\cA_2$ that do \emph{not} have $d_3$, or more, daughters in $\cA_3$, each of which have $d_4$, or more, in $\cA_4$ etc are dropped. The links of these dropped $\cA_2$ nodes to their parent nodes in $\cA_1$ are also dropped. Let $D_2\left( \cG_{(2)}(d_3 , \cdots , d_L) \right)$ be the total number of nodes in $\cA_2$ for the graph $\cG_{(2)}(d_3 , \cdots , d_L)$
\begin{comment}\footnote{ More specifically, this is the number of nodes in $\cA_2$ having $d_3$, or more daughters in $\cA_3$, each of which have $d_4$, or more, nodes in $\cA_4$ etc.}.
\end{comment}
 Then by definition,
\bea
\label{eq:D2forGd3d4etc}
	D_2\left( \cG_{(2)}(d_3 , \cdots , d_L) \right) = \left| \cS^{(2)}_{\left[ d_3 , d_4 , \cdots , d_L \right]} \right|. 
\eea
Combining Lemma \ref{lemma1} with (\ref{eq:D2forGd3d4etc}) gives
\bea
\label{eq:Lemma1layer12}
	\sum^{d^{max}_{2}}_{l_{2}=1} \left| \cS^{(1)}_{\left[l_2\right]} \left( \cG_{(2)}(d_3 , \cdots , d_L) \right) \right|  = D_2\left( \cG_{(2)}(d_3 , \cdots , d_L) \right) = \left| \cS^{(2)}_{\left[ d_3 , d_4 , \cdots , d_L \right]} \right|.
\eea
%Therefore,
%\bea
%	\sum^{d^{max}_{2}}_{l_{2}=1} \left| \cS^{(1)}_{\left[l_2\right]} \left( \cG_{(2)}(d_3 , \cdots , d_L) \right) \right| =  \left| \cS^{(2)}_{\left[ d_3 , d_4 , \cdots , d_L \right]} \right|.
%\eea
By construction, the quantity $\left| \cS^{(1)}_{\left[ l_2 \right]} \right|$ in the truncated graph $ \cG_{(2)}(d_3 , \cdots , d_L)$ is equivalent to $\left| \cS^{(1)}_{\left[ l_2 , d_3 , \cdots , d_L \right]} \right|$ in the original graph:
\bea 
	\left| \cS^{(1)}_{\left[l_2\right]} \left( \cG_{(2)}(d_3 , \cdots , d_L) \right) \right| =  \left| \cS^{(1)}_{\left[l_2 , d_3 , \cdots , d_L\right]} \right|.
\eea
Thus, applying this to (\ref{eq:Lemma1layer12}) gives
\bea
\label{eq:l2sum}
	\sum^{d^{max}_{2}}_{l_{2}=1}   \left| \cS^{(1)}_{\left[l_2 , d_3 , \cdots , d_L\right]} \right| = \left| \cS^{(2)}_{\left[ d_3 , d_4 , \cdots , d_L \right]} \right| .
\eea
Now consider a similar two-layer truncated graph $\cG_{(2)}( d_4 , \cdots , d_L)$ between $\cA_2$ and $\cA_3$. Denote the total number of nodes in $\cA_3$ in $\cG_{(2)}( d_4 , \cdots , d_L)$ by $D_3\left( \cG_{(2)}(d_4 , \cdots , d_L) \right)$. By definition, %Let the total number of nodes in $\cA_3$ having $d_4$, or more, daughters in $\cA_4$, each of which have $d_5$, or more, nodes in $\cA_5$ etc, is simply the total number of nodes in $\cA_3$ for $ \cG_{(2)}(d_4 , \cdots , d_L)$,
\bea
\label{eq:D3forTruncated}
	D_3\left( \cG_{(2)}(d_4 , \cdots , d_L) \right) =  \left| \cS^{(3)}_{\left[  d_4 , \cdots , d_L \right]} \right|.
\eea
Applying Lemma \ref{lemma1} to the truncated graph, together with (\ref{eq:D3forTruncated}), gives
\bea
\label{eq:cS2fortruncated}
	\sum^{d^{max}_{3}}_{l_3 = 1}\left| \cS^{(2)}_{\left[ l_3 \right]} \left(\cG_{(2)}( d_4 , \cdots , d_L\right)\right| = D_3 \left(\cG_{(2)}( d_4 , \cdots , d_L\right) =  \left| \cS^{(3)}_{\left[  d_4 , \cdots , d_L \right]} \right|.
\eea
By construction, $\cS^{(2)}_{\left[ l_3 \right]}$ for $ \cG_{(2)}(d_4 , \cdots , d_L) $ is equivalent to $\cS^{(2)}_{\left[ l_3 , d_4 , \cdots , d_L \right]}$ in the original graph, and thus, (\ref{eq:cS2fortruncated}) becomes
\bea
	\sum^{d^{max}_{3}}_{l_3 = 1} \cS^{(2)}_{\left[ l_3 , d_4 , \cdots , d_L \right]} =  \left| \cS^{(3)}_{\left[  d_4 , \cdots , d_L \right]} \right|.
\eea
Combining this result with (\ref{eq:l2sum}), we have 
\bea
	\sum^{d^{max}_{3}}_{l_3 = 1} \sum^{d^{max}_{2}}_{l_{2}=1}   \left| \cS^{(1)}_{\left[l_2 , l_3 , d_4 ,  \cdots , d_L\right]} \right|  =  \left| \cS^{(3)}_{\left[  d_4 , \cdots , d_L \right]} \right|.
\eea
At the $i$th step of this iteration, we can similarly define a truncated graph $\cG_{(2)}(d_{i+1} , \cdots , d_{L})$ in which nodes in $\cA_i$ \emph{not} having $d_{i+1}$, or more, daughters in $\cA_{i+1}$ etc are dropped. The links of these dropped nodes in $\cA_{i}$ to their parent nodes to $\cA_{i-1}$ are also dropped. Once again, we have
\bea
	D_{i} \left( \cG_{(2)}(d_{i+1} , \cdots , d_{L}) \right) = \left| \cS^{(i)}_{\left[ d_{i+1} , \cdots , d_{L} \right]} \right|.
\eea
Applying Lemma \ref{lemma1} to $ \cG_{(2)}(d_{i+1} , \cdots , d_{L})$ gives
\bea
\label{eq:Lemmaithite}
	\sum^{d^{max}_{i}}_{l_i=1} \left| \cS^{(i-1)}_{\left[ l_i \right]}   \left( \cG_{(2)}(d_{i+1} , \cdots , d_{L}) \right) \right|  =  D_{i} \left( \cG_{(2)}(d_{i+1} , \cdots , d_{L}) \right) = \left| \cS^{(i)}_{\left[ d_{i+1} , \cdots , d_{L} \right]} \right|.
\eea
Similarly, by construction, 
\bea
	 \left| \cS^{(i-1)}_{\left[ l_i \right]}   \left( \cG_{(2)}( d_{i+1} , \cdots , d_{L}) \right) \right| = \left|\cS^{(i-1)}_{\left[ l_i , d_{i+1} , \cdots , d_{L} \right]}\right|,
\eea	
which, when combined with (\ref{eq:Lemmaithite}), gives
\bea
	\sum^{d^{max}_{i}}_{l_i=1} \left|\cS^{(i-1)}_{\left[ l_i , d_{i+1} , \cdots , d_{L} \right]}\right| = \left| \cS^{(i)}_{\left[ d_{i+1} , \cdots , d_{L} \right]} \right|.
\eea
But using the results of previous iterations, $\cS^{(i-1)}_{\left[ l_i , d_{i+1} , \cdots , d_{L} \right]} $ can be written in terms of $\cS^{(1)}_{\left[ l_2 , l_3 , \cdots , l_i , \cdots ,  d_{L} \right]} $:
\bea
	\sum^{d^{max}_{i}}_{l_i=1} \cdots \sum^{d^{max}_{2}}_{l_2=1} \cS^{(1)}_{\left[ l_2 , \cdots , l_i , d_{i+1} , \cdots , d_{L} \right]} = \left| \cS^{(i)}_{\left[ d_{i+1} , \cdots , d_{L} \right]} \right|.
\eea
Letting $i = L-1$, we arrive at (\ref{eq:IdenSd2d3dL}). 

Finally, summing over $d_{L}$ from 1 to $d^{max}_{L}$ on both sides in (\ref{eq:IdenSd2d3dL}), and applying Lemma \ref{lemma1} on the RHS, we obtain (\ref{eq:CountingforDL}) and thus prove Proposition \ref{GenCountProp}. 
%Summing over $d_{i+1}$ from 1 to $d^{max}_{i+1}$, and applying Lemma \ref{lemma1} gives

%%%%%%%%%%%%%%%%%%%%%%%%%%%%%%%%%%%%%%%%%%%%%%%%%%%%%%%%%%%%%%%%%%%%%%%%%%%%%%%%%%%%%%%%%%%%%%%%%%%%%%%%%%%%%%%%%%%%%%%%%%%%%%%%%%%%%%%%%%%%%%%%%%%%%%%%%%%%%%%%%%%%%%%%%%%%%%%%%%%%%%%%%%%%%%%%%%%%%%%%%%%%%%%%%%%%%%%%%%%%

%%%%%%%%%%%%%%%%%%%%%%%%%%%%%%%%%%%%%%%%%%%%%%%%%%%%%%%%%%%%%%%%%%%%%%%%%%%%%%%%%%%%%%%%%%%%%%%%%%%%%%%%%%%%%%%%%%%%%%%%%%%%%%%%%%%%%%%%%%%%%%%%%%%%%%%%%%%%%%%%%%%%%%%%%%%%%%%%%%%%%%%%%%%%%%%%%%%%%%%%%%%%%%%%%%%%%%%%%%%%

\section{Proof by construction  for $L=1,2$ } 
\label{sec:PfConst} 

In this section we consider the case of one and two-layer  degeneracy graphs with any number of nodes, and prove for these cases that the projectors in the final layer can be written as linear combinations of the monomials specified by our main conjecture \eqref{mainconj}. 

\subsection{  The case $L=1$ }
\label{sec:const1} 

When $ \cA_1$ has dimension $ D_1$, it has a basis of $D_1$ projectors. These correspond to nodes in a 1-layer degeneracy graph, which we can label with $ \left[ D_1 \right] = \{ 1 , \cdots , D_1 \}$. In the $L=1$ case, $ \cA= \cA_1$ is generated by a single algebra element  $\cC_1$ with distinct eigenvalues. The nodes of the graph are labelled by the eigenvalues of $ \cC_1$ which are distinct.  To construct a specific projector in $\cA_1$ labeled by $J \in [ D_1 ] $, we use the well-known mathematical formula:
\bea\label{projprod}
	Q_{J}(\cC_1) = \prod\limits_{ \substack { i \in \left[ D_1 \right] \\ i \neq J }} \frac{\cC_1 - x^{(1)}_{i} }{ x^{(1)}_{J} - x^{(1)}_{i} }.
\eea
This product annihilates all projectors $P_{i}$ except for $i = J$. Further, when applied to $ P_i$, $\cC_1$ evaluates to $x^{(1)}_i$ and we have $Q_{J}(\cC_1) P_i = P_i$. We conclude that 
Thus,
\bea
	P_{J} = Q_{J}(\cC_1).
\eea
The degree in $\cC_1$ of $Q_{J}(\cC_1)$ is $\left[D_1\right] - 1$, and the monomials introduced in the construction of $\cA_1$ are
\bea
	\{ 1 , \cC_1 , \cC^2_1 , \cdots , \cC^{\left[D_1\right] - 1}_{1} \}. 
\eea
The formula \eqref{projprod} is used extensively  to show that character-distingishing conjugacy classes give non-linear generating sets \cite{KempRam,RS} and in  discussions of integrality of 2D TQFT constructions of representation theoretic quantities \cite{IDFCTS,RS,RowCol}. 
%We have verified in Mathematica from $n = 2$ to $n = 30$ that these monomials are independent. 

\subsection{The $L=2$ case : $\cA_1 \rightarrow \cA_2$ }

In this case, the CASS algebra $\cA$ as presented as minimally generated by a non-linear generating set of two  algebra elements $ \{ \cC_1 , \cC_2 \}$. 
 The set $\cS^{(1)}_{\left[ \vec{d} \right]}$ in definition (\ref{eq:DefSAL}) becomes $\cS^{(1)}_{\left[d_2\right]}$ - the set of vertices in $\cA_1$ with $d_2$, or more, daughters in $\cA_2$. 

We describe below an algorithm for construction of projectors in $\cA_2$ labeled by $b \in \left[ D_2 \right]$. The node $b$ also belongs to a block in $\left[D_2\right]$. Let
\bea
	b \in B^{(2)}_{a},
\eea
where $a \in \left[ D_1 \right]$ is the parent of $b$ and $a \in \cS^{(1)}_{\left[d_2\right]}$ and has exactly $d_2$  daughters in $\cA_2$. We describe an iterative procedure where the steps are labeled by $\left[ d_2  \right]$. We start at $\left[ d_2  \right] = \left[ 2 \right]$ and successively increase $d_2$ to its maximum $d^{max}_{2}$, which is the largest number of daughters of any vertex in $\cA_1$. Note that 
\bea
	d^{max}_{2} &=& \underset{d_2}{\hbox{Max}} \{ \left| \cS^{(1)}_{\left[d_2\right] } \right| > 0 \}.
\eea
At stage $\left[d_2\right]$ of this iterative procedure, we introduce monomials
\bea
\label{eq:MonosA1toA2staged3}
	\hbox{\bf Monom}(d_2) = \{ 1 , \cC_1 , \cdots , \cC^{|\cS^{(1)}_{\left[d_2\right]}| -1}_{1} \} \times  \cC^{d_2-1}_{2},
\eea
and construct all projectors labelled by $b$ whose parents in $\cA_1$ have $d_2$ daughters in $\cA_2$. We will use the following two projector-as-product operators to construct the projectors:
\bea
	Q_{a}(\cC_1) &=& \prod_{\substack{ a' \in \cS^{(1)}_{\left[d_2\right]} \\ a' \neq a }}  \frac{\cC_1 - x^{(1)}_{a'} }{ x^{(1)}_{a} - x^{(1)}_{a'} },\\
	Q_{b}(\cC_2) &=& \prod\limits_{\substack{ b' \in B^{(2)}_{a} \\  b' \neq b }}  \frac{\cC_2 - x^{(2)}_{b'} }{ x^{(2)}_{b} - x^{(2)}_{b'} }. 
\eea
The degree of $Q_{a}(\cC_1)$ is $\left| \cS^{(1)}_{\left[d_2\right] } \right|-1$, and the degree of $Q_{b}(\cC_2)$ is $d_2-1$. The $Q_{a}(\cC_1) $ annihilates all nodes in $\cA_1$ having $d_2$, or more, daughters in $\cA_2$ except for $a$, the parent of the node in $\cA_2$ we wish to construct. $Q_{a}(\cC_1)$ acting on a $P^{(1)}_{a'}$ for which $a' \in \left[ D_1 \right] \setminus \cS^{(1)}_{\left[d_2\right]}$ evaluates to $Q_{a}(x^{(1)}_{a'})P^{(1)}_{a'}$. Concretely, we can act with $Q_{a}(\cC_1)$ on the identity in $\cA_1$,
\bea
\label{eq:QT2A1toA2case}
	Q_{a}(\cC_1)1_{\cA_1} = Q_{a}(\cC_1)\sum\limits_{ a' \in \left[ D_1 \right]} P^{(1)}_{a'} = P^{(1)}_{a} + \sum\limits_{ a' \in \left[ D_1 \right] \setminus \cS^{(1)}_{\left[d_2\right]}} Q_{a}\left( x^{(1)}_{a'} \right)P^{(1)}_{a'}.
\eea
We can expand $P^{(1)}_{a}$ into its $d_2$ daughters in $\left[D_2\right]$:
\bea
	P^{(1)}_{a} = \sum\limits_{ \substack { b' \in B^{(2)}_{a} }} P^{(2)}_{b'}.
\eea
Acting with $Q_{b}(\cC_2)$ on equation (\ref{eq:QT2A1toA2case}) annihilates all of the $d_2$ daughters of $a$ except for $b$. Thus, (\ref{eq:QT2A1toA2case}) becomes
\bea
	Q_{b}(\cC_2)Q_{a}(\cC_1) = P^{(2)}_{b} + \sum\limits_{a' \in \left[ D_1 \right] \setminus \cS^{(1)}_{\left[d_2\right]}} Q_{a}\left( x^{(1)}_{a'} \right) Q_{b}(\cC_2) P^{(1)}_{a'}
\eea
The vertices $a' \in \left[D_1\right] \setminus \cS_{\left[d_2\right]}$ all have fewer than $d_2$ daughters in $\left[D_2\right]$. Thus, these corresponding projectors have already been constructed at earlier stages $d'_2 < d_2$. Thus, we can expand the $P^{(1)}_{a'}$ into their daughters in $\left[D_2\right]$, and the $Q_{b}(\cC_2)$ acting on these $\cA_2$ projectors may be evaluated:
\bea
	Q_{b}(\cC_2)Q_{a}(\cC_1) = P^{(2)}_{b} + \sum\limits_{ a' \in \left[D_1\right] \setminus \cS^{(1)}_{\left[d_2\right]}}  \sum\limits_{ \substack { b' \in B^{(2)}_{a'} }}  Q_{a}\left( x^{(1)}_{a'} \right) Q_{b}\left( x^{(2)}_{b'} \right) P^{(2)}_{b'}.\nonumber\\
\eea
We can now solve for the desired projector 
\bea
	P^{(2)}_{b}  = Q_{b}(\cC_2)Q_{a}(\cC_1)   -\sum\limits_{ a' \in \left[D_1\right] \setminus \cS^{(1)}_{\left[d_2\right]}}  \sum\limits_{ \substack { b' \in B^{(2)}_{a'} }}  \cV\left(P^{(2)}_{b'} \right).
\eea
where $\cV(P^{(2)}_{b'})$ denotes the subspace spanned by $P^{(2)}_{b'}$. The monomials introduced at stage $\left[d_2\right]$ are contained in $Q_{b}(\cC_2)Q_{a}(\cC_1)$ and are precisely those stated in (\ref{eq:MonosA1toA2staged3}).
 
 \vskip.2cm 
 
 \noindent 
\textbf{Remarks: }

\begin{itemize}
	
\item Centres of symmetric group algebras $ \cZ ( \mC( S_n ) )  $ for $n \in \{ 6,8, \cdots , 14 \}$ provide examples of this $ L=2$ case, by taking $ \cC_1 = T_2 $ to be the class sum of permutations with a single non-trivial cycle of length $2$, and $ \cC_2 = T_3 $ to be  the class sum of permutations with a single non-trivial cycle of length $3$. 
We have verified using mathematica that the monomials described here indeed span the centres of these algebras, showing that any projector can be written as a linear combination  of the monomials. 

\item A straightforward use of the projector as product formula \eqref{projprod} in the $L=2$ case, as in \cite{KempRam,RS}, only proves that monomials in the two generators provide spanning sets for  $ \cA= \cA_2$. 
 To illustrate this, we may consider an $\cA_1 \rightarrow \cA_2$ graph defined by $p_1 = (2,2,1,1)$. Here $D_1 = 4$ and $D_2 = 6$. Applying the projector-as-product formula to construct each projector in $\cA_2$ independently of any other $\cA_2$ projectors previously constructed will give 8 monomials in total $\{ 1 , \cC_1 , \cC^{2}_{1} , \cC^{3}_{1} \} \times \{ 1 , \cC_2 \}$. On the other hand  applying the construction of 
 \eqref{mainconj}, we have  $|\cS^{(1)}_{\left[1\right]}| = 4$, and $|\cS^{(2)}_{\left[2\right]}| = 2$ which yields $\cS(\cA_1 \rightarrow \cA_2 )$ to be the disjoint union of $\{ 1 , \cC_1 , \cC^2_1 , \cC^3_1  \}$ and $\{ 1 , \cC_1 \} \times \cC_2$.  

\end{itemize}

\section{ Graphs and determinants } 
\label{sec:GraphsDets} 

In this section, we use the monomial basis conjecture \eqref{mainconj} to write a formula for the matrix of expansion coefficients (equation \eqref{MatCoeffs}) of the monomials in terms of the projectors in $ \cA = \cA_L$, which correspond to nodes in the final layer of the degeneracy graph. An efficient way to prove \eqref{mainconj} would be to show the matrix has non-zero determinant when the eigenvalue labels of the degeneracy graphs satisfy the appropriate inequalities. Experimental study of the determinants for degeneracy graphs constructed systematically using the code in Appendix A shows that the determinant is indeed non-zero as expected from the conjecture. Further, there is an interesting factorised structure of the determinant related to the inequalities. This factorisation is formalised in the two conjectures \eqref{FactoredDet} and \eqref{DetConj2}. 

The determinant appearing in the change-of-basis matrix may be viewed as a layered generalisation of the Vandermonde determinant. Differences of eigenvalues appear as factors, with multiplicities determined by the combinatorial structure of the degeneracy graph. 

As explained in section \ref{sec:degengraphs}, the data $ \cD = ( L ; D_1 , D_2 , \cdots , D_L ; p_1 , c_2 , \cdots , c_{L-1} ) $ determines a degeneracy graph. $L$ is a positive integer and the graph is associated with the degeneracies of projectors in a sequence of algebras 
\bea 
\cA_1 \rightarrow \cA_2 \rightarrow \cdots \rightarrow \cA_{ L } \nnm 
\eea
generated by successively adding one generator at each stage. Thus 
\bea 
\cA_1  & = & \langle \cC_1 \rangle \subset \cA_L \cr
\cA_2 & = &  \langle \cC_1, \cC_2 \rangle \subset \cA_L  \cr  
& \vdots &  \cr 
\cA_{ i }   &  = & \langle \cC_1, \cC_2 , \cdots \cC_i  \rangle  \subset \cA_L  \cr 
& \vdots & \cr 
\cA_{ L-1 }   &  = & \langle \cC_1, \cC_2 , \cdots \cC_{ L-1}  \rangle  \subset \cA_L  \cr 
\cA_{ L } & = &  \langle \cC_1, \cC_2 , \cdots \cC_L   \rangle  \subset \cA_L 
\eea
The integers $D_i$ are increasing dimensions  
\bea 
D_1 < D_2 \cdots < D_L 
\eea 
The partition $p_1$ and the compositions $ c_2 , c_3 , \cdots , c_{ L-1} $, and the associated block 
decompositions of  $ [D_1] , [ D_2 ] , \cdots , [ D_L ] $, determine the degeneracy graph
as detailed in section \ref{sec:degrapdat}. Further, as elaborated in section \ref{sec:ProjEig}, each  node in layer $L$ corresponds to a projector in $ \cA_{ L } $ denoted $P^{(L)}_a$ with $ a \in [ D_L ] $. 
There are  eigenvalues  $ x^{(i)}_{ a} $ of $ \cC_{ i } $ : 
\bea 
&& \{ x^{ i}_a  :   a \in [ D_i ]  = \{ 1, 2 , \cdots, D_i\} \} \cr 
&&  \cC_i P^{(i)}_a = x^{(i)}_{ a}  P^{(i)}_a
\eea 
The eigenvalues obey the condition  \eqref{nondegleaves}, repeated here for convenience, 
\bea
 \hbox{ For  all  } a \in [ D_{i-1} ]  : ~~  b_1 , b_2 \in B^{(i)}_{ a } \hbox{ and }  b_1 \ne b_2   \implies x^{(i)}_{b_1}  \ne x^{(i)}_{b_2}  \nnm 
\eea

Given a node/projector at layer $L$ specified by   $ a_L \in  [D_L ] = \{ 1 , 2, \cdots , D_L \} $ we have an eigenvalue $ x^{(L)}_{a_L} $ of $ \cC_L $. The sequence of block decompositions specifies a sequence of ancestor projectors for $ P^{(L)}_{ a_L  } $ 
\bea 
P^{ (1)}_{ a_1} \rightarrow P^{ (2)}_{ a_2} \rightarrow \cdots \rightarrow  P^{(L-1)}_{ a_{L-1}}  \rightarrow P^{(L)}_{ a_L} 
\eea 
determined by the block decompositions : 
\bea 
 a_{L}  & \in &   B_{ a_{ L-1} }^{ (L-1)} \cr 
 a_{ L-1}  & \in & B_{ a_{ L-2} }^{ (L-2)} \cr 
% a_{L-2}    & \in &  B_{a_{ L-3}   }^{(L-3)  }  \cr 
 &  \vdots &  \cr 
 a_k   & \in   &  B_{a_{ k-1 }   }^{(k-1)  } \cr 
 &  \vdots &  \cr 
  a_{ 2}    & \in & B_{ {a_1}   }^{ (1)   } 
\eea
In turn we have a sequence of eigenvalues 
\bea 
( x^{(1)}_{ a_1} , x^{(2)}_{ a_2} , \cdots , x^{ (L)}_{ a_L } ) 
\eea
which uniquely determine the projector $P^{(L)}_{ a_L} $. 
In the graph picture $a_{L-1} \in [ D_{ L-1} ]  $ specifies the parent node in layer $ ( L-1) $ of the node $a_L  \in [ D_L ] $. In turn, $ a_{ L-2} \in [ D_{ L-2} ]  $ specifies the parent in layer $ ( L-2)$ of $ a_{ L-1} \in [ D_{ L-1} ]  $, and so forth walking backward to $a_1 \in [ D_1  ]$.

Any monomials in the generators, and in particular monomials specified by the exponents
$ \bm = ( m_1 , m_2 , \cdots ,m_L ) $  in the set $ \cS ( \cD  ) \equiv  \cS (\cA_1 \rightarrow \cdots \rightarrow \cA_L  ) $ have an expansion in projectors 
\bea\label{MonomEigs}  
\prod_{ k =1}^{ L } \cC_k^{ m_k }  
=   \sum_{ a_L  \in [ D_L ]  }   \left (     \prod_{ k=1}^{ L }    ( x_{ a_k  }^{(k)}  )^{ m_k    }       \right )  \cP^{ (L)}_{ a_L } 
\eea
This defines the matrix of coefficients  $ \cM $ of size $ ( D_L \times D_L  )  $ which expresses the basis of monomials in terms of the basis of projectors : 
\bea\label{MatLabels} 
\cM_{ a_L , \bm  }  :    a_{ L } \in [  D_L ]  , \bm \in \hbox { exponents of monomials in } \cS ( \cD  ) 
\eea
where the label $a_{ L } \in [ D_L ] $ runs over the projector basis  of $ \cA_L$ and the column index ${ \bm } $ runs over the monomials, and 
\bea\label{MatCoeffs} 
\cM_{ a_L , \bm  }  =  \left (     \prod_{ k=1}^{ L }    ( x_{ a_k  }^{(k)}  )^{ m_k    }       \right ) 
\eea 

\vskip.3cm 

\noindent 
{\bf Determinant Conjecture 1:  } 
The matrix $ \cM $  in \eqref{MatCoeffs} has a determinant 
\bea\label{FactoredDet}  
\det ( \cM ) = { \pm 1 }  \prod_{  i=1}^{ L }    \prod_{ a=1}^{ D_{ i-1 }  } \prod_{ p <  q \in B^{ (i)}_{ a }   }   ( x_p^{(i)} - x_q^{ (i)} )^{  {\rm Exponent}  (i,  p , q )  }   
\eea
with positive integer exponents, $ {\rm Exponent}  (i,  p , q )  > 0  $  for all $i \in  \{ 1, 2, \cdots , L \}  $, for all  $a \in [ D_{ i-1} ] $  and every pair $ p, q \in  B^{ (i)}_{ a } $. 

\vskip.2cm 

The product is over pairs of nodes at fixed level, which share the same parent : equivalently, the pairs at level  $i$ are in the same block as determined by the composition $ c_{ i-1} $ of 
$ [ D_i ] $ into $ D_{ i-1} $ parts. All the nodes in layer 1 are considered to have the same parent : this is naturally understood by extending the graph to a layer labelled $0$, which contains just one node and has edges connecting it to all the nodes in layer $1$.  In terms of the sequence of algebras we may consider extending to include $ \cA_0 = \mC $ which is spanned by complex multiples of the identity element $ \mathbf{1}$. In the formula \eqref{FactoredDet}, $ D_0$  is defined as $1$ and
$ B^{ (1)}_{ a=1 }  $ is defined to include all the nodes in layer 1 of the graph. 

\vskip.3cm 
 
\noindent 
{\bf Determinant Conjecture 2:}
 The exponents $ { \rm Exponent}  ( i , p , q )$  in \eqref{FactoredDet} are given by: 
\begin{align}\label{DetConj2}  
 { \rm Exponent}  ( i , p , q )  &=  1 \cr 
 ~~\hbox{ for } ~~ i & = L \cr 
  { \rm Exponent}  ( i , p , q )  &=   \hbox { Number monomials containing $ \cC_1$ in the truncated graph  }\cG_{ \trunc } ( i , p , q  ) \cr 
  ~~ \hbox{ for } ~~ i &  \in  \{ 1 , \cdots , L-1  \} 
\end{align}
The truncated graph $ \cG_{ \trunc } ( i , p , q  ) $   drops the layers $ j < i $, and at layer $i$, keeps only the two nodes $ p,q \in [ D_i ] $.  Further, at layers $j$ in the range $ i+1 \le j \le L $, it only keeps the descendants of the two nodes $p,q$ along with all the edges linking these descendants back to $p,q$. 

\vskip.2cm 

We can equivalently express conjecture 2 by saying that the exponent is 
the number of values of  vectors $ [  d_{ i+1} , \cdots , d_{ L } ] $ for which 
\bea  
| \cS^{(i)}_{ [  d_{ i+1} , \cdots , d_L ] }  ( \cG_{ \trunc } ( i , p , q  ) )  | = 2
\eea
The equality holds for   $ [ d_{ i+1}, \cdots , d_L ] = [ 1,1, \cdots , 1  ]$, so the exponent is at least $1$. 

We give code written in sagemath \cite{sagemath} which performs the following tasks:  
\begin{itemize} 

\item Given data $L$ , $ D = [ D_1 , \cdots , D_L ] $ obeying $ D_1 < D_2 < \cdots < D_L $, and partitions/compositions $ \{ p_1 , c_2 , \cdots , c_{ L-1} \} $, produces the graph, the monomial basis, the transformation matrix, and the factored determinant.  

\item Verifies conjecture 1. 

\item Verifies Conjecture 2. 

\end{itemize} 

A guide to the code is in the appendix \ref{sec:AppCode}.

\section{ Properties of the monomial basis } 
\label{sec:Props} 

In this section we describe some properties  the monomial basis $ \cS ( \cA_1 \rightarrow \cdots \rightarrow \cA_L ) $ given in section \ref{sec:monombasis}.

\subsection{ Partitions, compositions and symmetries } 

In section \ref{sec:degrapdat}, we described how the data of $ \cD = \{ L ; D_1 , D_2 , \cdots , D_L ; p_1 , c_2 , \cdots , c_{ L-1} \} $ determines the degeneracy structure associated with an ordered sequence of generators $ \{ \cC_1 , \cC_2 , \cdots , \cC_{ L } \} $ of an algebra. 
We chose $p_1$ to belong to the set of partitions of $D_2$ with $D_1$ parts. In constructing the graphs we can choose, without loss of generality, to arrange the parts of $p_1$ to be in weakly increasing order as we go up along the first layer (as illustrated in the graphs in the Appendix 
\ref{sec:AppCode}). These parts determine the number of daughters of each of the $D_1$ nodes. 
The data $c_2$ belongs to the set of compositions of $D_3 $ into $D_2$ parts, which is larger than the set of partitions of $D_3$ with $D_2$ parts, since compositions are counted as distinct when they have the same parts in a different order. The choice of compositions $ c_2, c_3 , \cdots , c_{ L-1}$ to describe  the node connections between successive layers beyond the second layer ensures that we produce the full set of degeneracy graphs. We can start with the larger data of a composition $c_1$ for the node connections between the first two layers as well, but we would then be producing an obvious redundancy of equivalent graphs related by re-ordering the nodes in the first layer. The data including compositions $\{ c_2 , \cdots , c_{ L-1} \} $ does have further more subtle redundancies associated with orbits of wreath product groups of the kind 
$S_{ r_1 } [ S_{ r_2 } ] $,  iterated wreath products $ S_{ r_1} [ S_{ r_2 } [ S_{ r_3 } ] ] $  
and higher iterations. A systematic counting of the degeneracy graphs at each layer, which takes into account these redundancies is an interesting problem for the future.

\subsection{ Order ideal property } 

If a monomial $ \cC_1^{ a_1} \cC_{2}^{ a_2} \cdots \cC_{ L }^{ a_L } $ appears in the  proposed monomial basis set, then all the monomials $ \cC_1^{ b_1 } \cC_{2}^{ b_2} \cdots \cC_{ L }^{ b_L  } $ with $ 1\le b_1 \le a_1 , 1 \le b_2  \le a_2 , \cdots , 1 \le b_L  \le a_L $ also  appear in the  basis. 

First consider the case where $ a_2=b_2, \cdots  , a_L = b_L$, and $ b_1 < a_1 $. In this case the above property is evident because  of the definition in  \eqref{eq:MonoBasisforAL}. 

The next case is where $ [ b_2 , \cdots , b_L] \le [ a_2 , \cdots , a_L ]$ without complete equality. In this case, the property follows because the subsets $ S^{(1)}_{ \vec d } \subset \cA_1$, defined in section \ref{sec:defS1},  obey the inclusions: 
\bea 
 &&   1 \le b_2  \le a_2 , \cdots , 1 \le b_L  \le a_L  
 \cr
 && \implies  ~~~    \cS^{(1)}_{ [ a_2 , \cdots , a_{ L}] } \subseteq \cS^{(1)}_{[ b_2 , \cdots , b_L] } \cr 
 && 
 \eea 
 When we decrease one or more of the sequence $ \{  a_2 , \cdots , a_L \} $ to obtain the $ \{ b_2, \cdots, b_L \} $, the defining condition of the $S^{(1)}_{ [ \vec b ] } $ is weaker than the one $\cS^{(1)}_{ [ \vec a ] }$, as the lower bounds on the daughter degeneracies are being relaxed. It follows that 
 \bea 
  | \cS^{(1)}_{[  a_2 , \cdots , a_{ L}] } |  \le | \cS^{(1)}_{ [  b_2  , \cdots , b_L]  } |
\eea

This can be phrased as a downward closed property (also called an order ideal property) in the set $ {\mathbb{N}}^{ \times L } $, which consists of tuples of non-negative integers $ ( a_1 , \cdots , a_L )$. A partial order $ \le $ on the set is defined by 
\bea 
a \le b  ~~\hbox{ if and only if } ~~ a_i \le b_i ~~ \hbox{ for all  } ~~~ i 
\eea
A subset $ S \subset {\mathbb{N}}^{ \times L } $ is said to be downward closed (or an order ideal) if 
$ a \in S $ implies that $ b \in S $ for all $ b \le a$.

The order ideal property suggests that an alternative way to study centres of group algebras will be to look at them as quotients of a polynomial ring $ \mathbb{R} [ \cC_1 , \cdots , \cC_L ]  $ 
by an ideal $ \cI $ using results in computational algebraic geometry \cite{HH2011,CLO2015}. 
The ideal $\cI$ will be generated by the monomials corresponding to the complement in 
$ \mathbb{N}^{ L }$ of the exponent set defined above. An interesting future direction is to realise the finite algebra 
$\cA$ as a quotient of the polynomial ring 
$\mathbb{R} [ \cC_1 , \cdots , \cC_L ]$ by a monomial ideal $I$,
and to determine generators and relations for $I$.

\subsection{Dependence on the order of generators }

The definition of the sequence of sub-algebras of $\cA =  \cA_L$ we described in \ref{sec:degengraphs}, 
depends on a choice of ordering of the $L$ generators of the minimal generating set. The degeneracy graph, and resulting monomial basis,  likewise depends on this choice. The first layer has a number of nodes equal to the number of eigenvalues of $ \cC_1$, the second has a number of nodes equal to the number of distinct  ordered eigenvalue pairs for $ \{ \cC_1 , \cC_2 \} $. 
The number of nodes in the final layer is independent of the choice of ordering of the generators, since it is equal to the dimension of $ \cA_L$. The monomial basis construction of 
\ref{sec:monombasis} is defined for any choice of ordering and the counting proof of section \ref{sec:Counting} holds for any choice. 
However, the final algebra and its primitive projectors are independent of this choice, and each ordering yields a valid basis construction.
A systematic account of the $S_L$ action on the degeneracy graphs obtained from different choices of ordering of the generators, and the characterisation of the complete space of combinatorial  $S_L$ invariants would be interesting. 

\section{Applications of the monomial basis. } 
\label{sec:Apps}

In this section we describe applications of the monomial basis given in section \ref{sec:monombasis}.
The initial motivations for this paper described in the introduction came from centres of group algebras, in particular symmetric group algebras, which inform correlators of $U(N)$ gauge theory, in particular $ \cN=4$ SYM. The first application we consider in section 
\ref{sec:centressymm} relates directly to this motivational example. 

As we have seen,  the question of the construction of projectors in terms of a  non-linear generating sets,  which we were led to, is naturally tackled in the wider context of commutative associative semi-simple algebras. Further interesting instances of these are 
maximally commutative sub-algebras of {\it non-commutative} associative algebras. Sections \ref{sec:JM} and \ref{sec:MUandMultMat} describe applications in this more general context.

\subsection{ Examples from centres of symmetric group algebras } 
\label{sec:centressymm}

The monomial bases can be used to start with a minimal generating set for the centre of the symmetric group algebra $ \cZ( \mC ( S_n) ) $. Physically interesting examples of such generating sets are the cycle operators of increasing cycle length 
$ \{ T_2 , T_3 , \cdots , T_{ k_*(n)} \} $. We use the monomial basis specified in \eqref{mainconj} to obtain the expansion of the monomials in terms of projectors using the known characters of this small subset of conjugacy classes. By inverting the matrix of expansion coefficients, we express the projectors in terms of the generating set. The characters for more general conjugacy classes beyond the generating set can be read off from the eigenvalues obtained by applying the more general class sums to the projectors.

We illustrate this for the cases of $n = 5$ and $n = 10$. The sequence of subalgebras are $\cA_1 (L=1)$ and $\cA_1 \rightarrow \cA_2, (L=2)$ respectively. From the data given in the degeneracy graph, we give the monomials $\cS(\cA_1)$ for the $n=5$ case, and $\cS(\cA_1 \rightarrow \cA_2)$ for the $n=10$ case. We can then expand the monomials in terms of the projector basis thus generating the matrix of coefficients $\cM$. This matrix has maximal rank and is therefore invertible. Inverting the system of equations then allows us to express the projector basis in terms of the monomial basis. As a last check of the conjecture and algorithm, we verify various eigenvalue equations of the form (\ref{eigvaleqforClassSum}),
\bea 
\cC_i P^{(i)}_a = x^{ (i)}_a P^{(i)}_a,
\eea
the class sums should satisfy.

\subsubsection{$n = 5$}

For this example, the data specifying the degeneracy graph is $\cD = (L;D_1) = (1;7)$. The degeneracy graph is a single layer with seven nodes in $\cA_1$. 
The monomials are
\bea
	\cS(\cA_1) = \{ 1 , \cC_1 , \cC^{2}_{1} , \cdots , \cC^{6}_{1} \}.
\eea
The matrix expressing the monomials in the projector basis is
\bea
	&&\left( 1, \cC^{}_{1} , \cC^{2}_{1}  , \cC^{3}_{1}  , \cC^{4}_{1}  ,\cC^{5}_{1} , \cC^{6}_{1}  \right)  = \left( P_{(5)} , P_{(4,1)} , P_{(3,2)} , P_{(3,1,1)} , P_{(2,2,1)} , P_{(2,1,1,1)} , P_{(1,1,1,1,1)}\right) \cM \nonumber\\
	&&\nonumber\\
	\hbox{where }&&\cM = \left(
\begin{array}{ccccccc}
 1 & 10 & 100 & 1000 & 10000 & 100000 & 1000000 \\
 1 & 5 & 25 & 125 & 625 & 3125 & 15625 \\
 1 & 2 & 4 & 8 & 16 & 32 & 64 \\
 1 & 0 & 0 & 0 & 0 & 0 & 0 \\
 1 & -2 & 4 & -8 & 16 & -32 & 64 \\
 1 & -5 & 25 & -125 & 625 & -3125 & 15625 \\
 1 & -10 & 100 & -1000 & 10000 & -100000 & 1000000 \\
\end{array}
\right).
\eea
Its inverse expresses the projectors in terms of the monomials $\cS(\cA_1)$
\bea
	&&\left(  P_{(5)} , P_{(4,1)} , P_{(3,2)} , P_{(3,1,1)} , P_{(2,2,1)} , P_{(2,1,1,1)} , P_{(1,1,1,1,1)} \right) = \left( 1, \cC^{}_{1} , \cC^{2}_{1}  , \cC^{3}_{1}  , \cC^{4}_{1}  ,\cC^{5}_{1} , \cC^{6}_{1}  \right)\cM^{-1} \nonumber\\
	&& \nonumber\\
	\hbox{where }&&\cM^{-1} = \left(
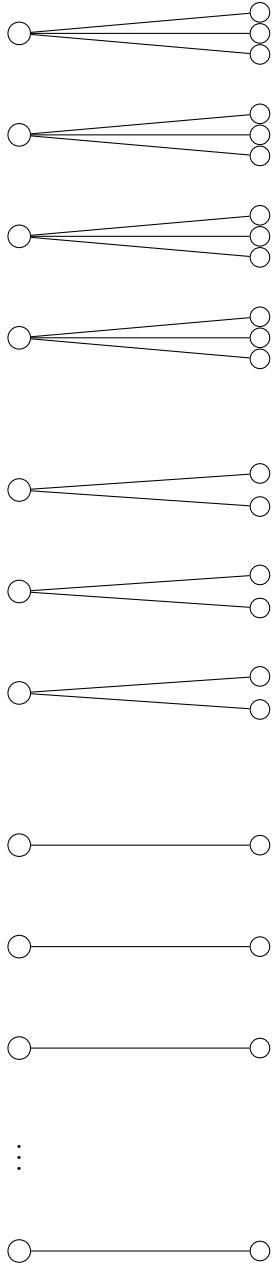
\begin{array}{ccccccc}
 0 & 0 & 0 & 1 & 0 & 0 & 0 \\
 \frac{1}{1440} & -\frac{8}{315} & \frac{625}{2016} & 0 & -\frac{625}{2016} & \frac{8}{315} & -\frac{1}{1440}
   \\
 \frac{1}{14400} & -\frac{8}{1575} & \frac{625}{4032} & -\frac{3}{10} & \frac{625}{4032} & -\frac{8}{1575} &
   \frac{1}{14400} \\
 -\frac{29}{144000} & \frac{52}{7875} & -\frac{125}{8064} & 0 & \frac{125}{8064} & -\frac{52}{7875} &
   \frac{29}{144000} \\
 -\frac{29}{1440000} & \frac{52}{39375} & -\frac{125}{16128} & \frac{129}{10000} & -\frac{125}{16128} &
   \frac{52}{39375} & -\frac{29}{1440000} \\
 \frac{1}{144000} & -\frac{1}{15750} & \frac{1}{8064} & 0 & -\frac{1}{8064} & \frac{1}{15750} &
   -\frac{1}{144000} \\
 \frac{1}{1440000} & -\frac{1}{78750} & \frac{1}{16128} & -\frac{1}{10000} & \frac{1}{16128} &
   -\frac{1}{78750} & \frac{1}{1440000} \\
\end{array}
\right).\nonumber\\
\label{eq:InverseMneq5}
\eea
We can verify the eigenvalue equations 
\bea
	\cC_\lambda P_{R} = \widehat \chi_R(\cC_\lambda) P_R.
\eea
Choosing $\cC_\lambda = \cC_5$, which the sum of permutations with cycle length 5, the normalised characters can be calculated to be \cite{CGS}
\bea
	&&\{ \widehat \chi_{(5)}(\cC_5) , \widehat \chi_{(4,1)}(\cC_5) , \widehat \chi_{(3,2)}(\cC_5) , \widehat \chi_{(3,1,1)}(\cC_5) , \widehat \chi_{(2,2,1)}(\cC_5)  , \widehat \chi_{(2,1,1,1)}(\cC_5) , \widehat \chi_{(1^5)}(\cC_5)   \} \nonumber\\
	&& \hspace{150pt} = \{24, -6, 0, 4, 0, -6, 24\} 
\eea
Using (\ref{eq:InverseMneq5}), we have verified
\bea
	\cC_5 P_{(5)} &=& 24 P_{(4,1)},  \hspace{20pt }\cC_5 P_{(4,1)} = -6 P_{(4,1)}, \hspace{20pt }\cC_5 P_{(3,2)} = 0 P_{(3,2)}\\
	\cC_5 P_{(3,1,1)} &=& 4 P_{(4,1)}\hspace{20pt }\cC_5 P_{(2,2,1)} = 0 P_{(2,2,1)}, \hspace{20pt }\cC_5 P_{(2,1,1,1)} = -6 P_{(2,1,1,1)}\\
	\cC_5 P_{(1,1,1,1,1)} &=& 24 P_{(1,1,1,1,1)}.
\eea

\subsubsection{$n = 10$}

For $n = 10$, the data specifying the degeneracy graph is 
\bea
	\cD &=& (2;31,42; p_1)\\ 
	p_1 &=& (3, 3, 3, 3, 2, 2, 2, 1, 1, 1, 1, 1, 1, 1, 1, 1, 1, 1, 1, 1, 1, 1, 1, 1, 1, 1, 1, 1, 1, 1, 1).
\eea
The degeneracy graph is given below.

\begin{center}
\begin{tikzpicture}[
  xnode/.style={circle, draw, minimum size=3mm, inner sep=0pt},
  ynode/.style={circle, draw, minimum size=2.6mm, inner sep=0pt},
  line width=0.4pt
]

\def\xsep{1.35}
\def\childsep{0.28}
\def\xshift{3.2}

%------------------------------------------------
% Layer 1  (TOP = most children)
%------------------------------------------------

% ---- 3 children (TOP)
\node[xnode] (x1) at (0,12*\xsep) {};
\node[xnode] (x2) at (0,11*\xsep) {};
\node[xnode] (x3) at (0,10*\xsep) {};
\node[xnode] (x4) at (0,9*\xsep)  {};

% ---- 2 children
\node[xnode] (x5) at (0,7.5*\xsep) {};
\node[xnode] (x6) at (0,6.5*\xsep) {};
\node[xnode] (x7) at (0,5.5*\xsep) {};

% ---- 1 child
\node[xnode] (x8)  at (0,4*\xsep)  {};
\node[xnode] (x9)  at (0,3*\xsep)  {};
\node[xnode] (x10) at (0,2*\xsep)  {};

% dots
\node at (0,1*\xsep) {$\vdots$};

% final node
\node[xnode] (xN) at (0,0) {};

%------------------------------------------------
% Layer 2
%------------------------------------------------

% 3 children
\foreach \i/\pos in {1/12,2/11,3/10,4/9}{
  \node[ynode] (a\i1) at (\xshift,{\pos*\xsep-\childsep}) {};
  \node[ynode] (a\i2) at (\xshift,{\pos*\xsep}) {};
  \node[ynode] (a\i3) at (\xshift,{\pos*\xsep+\childsep}) {};
}

% 2 children
\foreach \i/\pos in {5/7.5,6/6.5,7/5.5}{
  \node[ynode] (b\i1) at (\xshift,{\pos*\xsep-0.22}) {};
  \node[ynode] (b\i2) at (\xshift,{\pos*\xsep+0.22}) {};
}

% single children
\node[ynode] (c8)  at (\xshift,{4*\xsep}) {};
\node[ynode] (c9)  at (\xshift,{3*\xsep}) {};
\node[ynode] (c10) at (\xshift,{2*\xsep}) {};
\node[ynode] (cN)  at (\xshift,{0}) {};

%------------------------------------------------
% Edges
%------------------------------------------------

\foreach \i in {1,2,3,4}{
  \draw (x\i) -- (a\i1);
  \draw (x\i) -- (a\i2);
  \draw (x\i) -- (a\i3);
}

\foreach \i in {5,6,7}{
  \draw (x\i) -- (b\i1);
  \draw (x\i) -- (b\i2);
}

\draw (x8) -- (c8);
\draw (x9) -- (c9);
\draw (x10) -- (c10);
\draw (xN) -- (cN);

\end{tikzpicture}
\end{center}
The cardinalities of the $\cS^{(1)}_{\left[d_2\right]}$ sets are
\bea
	\left| \cS^{(1)}_{\left[1\right]}  \right| = 31, \hspace{10pt} \left| \cS^{(1)}_{\left[2\right]}  \right| = 7,  \hspace{10pt} \left| \cS^{(1)}_{\left[3\right]}  \right| = 4. 
\eea
Thus, the monomial basis is given by 
\bea
	\cS(\cA_1 \rightarrow \cA_2) =  \{ 1 , \cC_1 , \cdots , \cC^{30}_{1} \} \,\sqcup \, \{ 1, \cC_1 , \cC^{2}_{1} , \cdots \cC^{6}_{1} \}\times \cC_2 \,\sqcup \, \{ 1, \cC_1 , \cC^{2}_{1} ,\cC^{3}_{1} \}\times \cC^{2}_2 .  
\eea
The $42\times 42$ matrix $\cM$, and its inverse, are too large to display explicitly here. Instead we give some examples of projectors constructed from the monomial basis. First, we give $P_{(6,2,2)}$: 
\bea
P_{(6,2,2)}&=&
-\frac{\;\cC_1^{30}}{5522254426354051088989949421158400}-\frac{\;\cC_1^{29}}{502023129668550098999086311014400} \nonumber\\
&& +\frac{
   1579 \;\cC_1^{28}}{1380563606588512772247487355289600}+\frac{1579
   \;\cC_1^{27}}{125505782417137524749771577753600} \nonumber\\
&& -\frac{2740559
   \;\cC_1^{26}}{920375737725675181498324903526400}-\frac{2740559
  \;\cC_1^{25}}{83670521611425016499847718502400} \nonumber\\
&& +\frac{5878925021
   \;\cC_1^{24}}{1380563606588512772247487355289600}+\frac{5878925021
   \;\cC_1^{23}}{125505782417137524749771577753600} \nonumber\\
&& -\frac{180094781021
   \;\cC_1^{22}}{48019603707426531208608255836160} - \frac{180094781021
   \;\cC_1^{21}}{4365418518856957382600750530560} \nonumber\\
&& +\frac{5387449482913
   \;\cC_1^{20}}{2501021026428465167115013324800}+\frac{5387449482913
   \;\cC_1^{19}}{227365547857133197010455756800} \nonumber\\
&& -\frac{154417204081099
   \;\cC_1^{18}}{185834379672703294150960742400}-\frac{154417204081099
   \;\cC_1^{17}}{16894034515700299468269158400}\nonumber\\
&&+\frac{413132738278397
   \;\cC_1^{16}}{1896269180333707083173068800}+\frac{413132738278397
   \;\cC_1^{15}}{172388107303064280288460800}\nonumber\\
&&-\frac{18142803896872086329
   \;\cC_1^{14}}{468027326583104592676493721600}-\frac{18142803896872086329
   \;\cC_1^{13}}{42547938780282235697863065600}\nonumber\\
&&+\frac{132908846628741353
   \;\cC_1^{12}}{28806346631289236348731392}+\frac{132908846628741353
   \;\cC_1^{11}}{2618758784662657849884672} \nonumber\\
&& -\frac{2351154314430344353675
   \;\cC_1^{10}}{6587051263021472045076578304}-\frac{2351154314430344353675
   \;\cC_1^9}{598822842092861095006961664} \nonumber\\
&& + \frac{47964038211297386431875
   \;\cC_1^8}{2805595908323960315495579648}+\frac{47964038211297386431875
   \;\cC_1^7}{255054173483996392317779968} \nonumber\\
&& -\frac{15271513417075733578125
   \;\cC_1^6}{32718319630600120297324544}-\frac{15271513417075733578125
   \;\cC_1^5}{2974392693690920027029504} \nonumber\\
&& +\frac{520566788125466015625
   \;\cC_1^4}{83465101098469694636032}+\frac{520566788125466015625
   \;\cC_1^3}{7587736463497244966912} \nonumber\\
&& -\frac{774196398193359375
   \;\cC_1^2}{26976438622646960128}-\frac{774196398193359375 \;\cC_1}{2452403511149723648}
\eea
Note that only $\cC_1$ is needed to construct $P_{(6,2,2)}$. Next, we give an example of a projector needing $\cC_1$ and $\cC_2$
\[
\resizebox{1.15\textwidth}{!}{$
\begin{aligned}
P_{(3,3,2,2)} &= -\frac{584161925369 \;\cC_1^{30}}{116000036232563518368057498009600000000000000}+\frac{162738322864283
   \;\cC_1^{29}}{5213373056966354696941555553402880000000000000}\\&
   +\frac{1458204377509648361
   \;\cC_1^{28}}{45617014248455603598238611092275200000000000000}-\frac{1808374066902925993
  \;\c C_1^{27}}{9123402849691120719647722218455040000000000000}\\&-\frac{1529055968126424940753
   \;\cC_1^{26}}{18246805699382241439295444436910080000000000000}+\frac{145927728789353295893
   \;\cC_1^{25}}{280720087682803714450699145183232000000000000}\\&+\frac{2967822301797334728181
   \;\cC_1^{24}}{24538469203042282731704470732800000000000000}-\frac{526819064214386540849179
  \;\c C_1^{23}}{701800219207009286126747862958080000000000000}\\&-\frac{65783412141385967677058777
  \;\c C_1^{22}}{610261060180008074892824228659200000000000000}+\frac{7428014303938786707475691
   \;\cC_1^{21}}{11095655639636510452596804157440000000000000}\\&+\frac{31136859629840187523420121
   \;\cC_1^{20}}{495341769626629930919500185600000000000000}-\frac{2979787169441349841770638021
   \;\cC_1^{19}}{7628263252250100936160302858240000000000000}\\&-\frac{198448408671424153189976816759
   \;\cC_1^{18}}{8029750791842211511747687219200000000000000}+\frac{188712042936430154615687495963
   \;\cC_1^{17}}{1228079532869985290031999221760000000000000}\\&+\frac{4073891259383876167486340516267
  \;\c C_1^{16}}{614039766434992645015999610880000000000000}-\frac{5067610981455885815362017590251
   \;\cC_1^{15}}{122807953286998529003199922176000000000000}\\&-\frac{46115916608792769774331873034852917
   \;\cC_1^{14}}{37958821925072272600989066854400000000000000}+\frac{6933195978780175713977264242252741
   \;\cC_1^{13}}{917685804781967029914021396480000000000000}\\&+\frac{1755894547812433824894454063234838279
   \;\cC_1^{12}}{11735789618846309132554312089600000000000000}-\frac{286527161632496771228598810461798249
   \;\cC_1^{11}}{308212656656569734794355671040000000000000}\\&-\frac{1483392354700094092421079169201878701
  \;\c C_1^{10}}{123285062662627893917742268416000000000000}+\frac{149938460197325026472362816238098457
   \;\cC_1^9}{2009089910057639752733577707520000000000}\\&+\frac{1541236690341222297542614164599693669
   \;\cC_1^8}{2567170440629206350715127070720000000000}-\frac{5721004596927777411817644429894828911
   \;\cC_1^7}{1540302264377523810429076242432000000000} \\& -\frac{20989972375051968467019009137832214993
  \;\cC_1^6}{1232241811502019048343260993945600000000}+\frac{7480202528492092030739772980078357617
  \;\cC_1^5}{71196193553449989459832857427968000000} \\& +\frac{189878432079024987353938557400197739
  \;\cC_1^4}{809047654016477152952646107136000000}-\frac{13583679524069232687062445130593613
  \;\cC_1^3}{9417485919768517124316515532800000} \\& - \frac{33577691389045313943369938486543
 \;\cC_1^2}{30378986837962958465537146880000}+\frac{1227550266059326111978047
 \;\cC_1}{182119734641790413701120} \\& -\frac{3 \;\cC_2^2}{12800}-\frac{9 \;\cC_1 \;\cC_2}{5120}+\frac{3
  \;\cC_1 \;\cC_2^2}{64000}+\frac{53 \;\cC_1^2 \;\cC_2}{128000}+\frac{\;\cC_1^2 \;\cC_2^2}{38400}+\frac{61
  \;\cC_1^3 \;\cC_2}{320000} \\& -\frac{\;\cC_1^3 \;\cC_2^2}{192000}-\frac{137
  \;\cC_1^4 \;\cC_2}{2880000}+\frac{\;\cC_1^5 \;\cC_2}{1920000}+\frac{\;\cC_1^6 \;\cC_2}{5760000}+\frac{27}{512}
\end{aligned}
$}
\]
We have also verified that summing over all projectors constructed in terms of $\cS(\cA_1 \rightarrow \cA_2)$ using $\cM^{-1}$ satisfies
\bea
	\sum_{R\vdash 10} P_{R} = 1.
\eea
We may also verify the eigenvalue equations. For example, the normalized characters for the class sums $\cC_1 , \cC_2$ and $\cC_{(2,2)}$, taken over the $S_{10}$ irreps, are
\bea
	\widehat{\chi}_{R}(\cC_1) &\rightarrow& \{45, 35, 27, 25, 21, 18, 15, 17, 13, 11, 9, 5, 15, 10, 7, 5, 3, 0, \
-5, 5, 3, 3, 0, -3, -3,\nonumber\\
&&  -5, -9, -15, -3, -5, -7, -11, -10, -13, -18,-25, -15, -17, -21, -27, -35, -45\}\\
\widehat{\chi}_{R}(\cC_2) &\rightarrow& \{240, 160, 96, 100, 48, 51, 60, 16, 16, 16, 24, 40, 0, -5, -8, 0, 0, \
15, 40, -20, -12,\nonumber\\
&& -24, -15, 0, -12,  0, 24, 60, -24, -20, -8, 16, -5, 16, 51, 100, 0, 16, 48, 96, 160, 240\}\\
\widehat{\chi}_{R}(\cC_{(2,2)}) &\rightarrow& \{630, 350, 198, 140, 126, 63, 0, 98, 38, 14, -18, -70, 90, 35, 14, \
-10, -18, -45, \nonumber\\
	&&-70, 20, 0, 18, 0, -18, 0, -10, -18, 0, 18, 20, 14, 14, 35, 38, 63, 140, 90, 98, 126, 198, 350, 630\}. \nonumber\\
\eea
Note that amongst the eigenvalues $x^{(1)}_{i}$ for $\widehat{\chi}_{R}(\cC_1)$, $3,5,-3,-5$ each have degeneracy 3, while $0,15,-15$ each have degeneracy 2. All other eigenvalues have degeneracy 1. We have verified that
\bea
	\cC_1 P_{R} = \widehat{\chi}_{R}(\cC_1) P_{R} , \hspace{10pt} \cC_2 P_{R} = \widehat{\chi}_{R}(\cC_2) P_{R}  , \hspace{10pt} \hbox{and} \hspace{10pt} \cC_{(2,2)} P_{R} = \widehat{\chi}_{R}(\cC_{(2,2)}) P_{R},
\eea
for all $R$. In particular, for $(6,2,2)$ and $(3,3,2,2)$ that we explicitly presented above,
\bea
\{ \widehat \chi_{(6,2,2)}(\cC_1 ) ,  \widehat \chi_{(6,2,2)}(\cC_2 )  , \widehat \chi_{(6,2,2)}(\cC_{(2,2)} )  \} &=& \{ 11, 16 , 14 \},\\
\{ \widehat \chi_{(3,3,2,2)}(\cC_1 ) ,  \widehat \chi_{(3,3,2,2)}(\cC_2 )  , \widehat \chi_{(3,3,2,2)}(\cC_{(2,2)} )  \} &=& \{ -5, -20 , 20 \}.
\eea
We have checked that
\bea
	&&\cC_1 P_{(6,2,2)} = 11 P_{(6,2,2)} , \hspace{10pt}\cC_2 P_{(6,2,2)} = 16 P_{(6,2,2)} , \hspace{10pt} \hbox{and} \hspace{10pt} \cC_{(2,2)} P_{(6,2,2)} = 14 P_{(6,2,2)} ,\\
	&&\cC_1 P_{(3,3,2,2)} = -5 P_{(3,3,2,2)} , \hspace{10pt}\cC_2 P_{(3,3,2,2)} = -20 P_{(3,3,2,2)} , \hspace{10pt} \hbox{and} \hspace{10pt} \cC_{(2,2)} P_{(3,3,2,2)} = 20 P_{(3,3,2,2)}.\nonumber\\
\eea

\subsection{ Maximal commuting sub-algebras for $ \mC ( S_n) $ : Jucys-Murphy elements } 
\label{sec:JM}

Well-studied instances  of maximally commuting sub-algebras of non-commutative semisimple algebras are the Jucys-Murphy algebras of symmetric group algebras $ \mC( S_n ) $. 
In applications to gauge invariant operators in large-$N$ gauge theories,  Jucys--Murphy elements in symmetric group algebras have been used for example in \cite{EHS,DIS}. 

 There are elegant formulae for the eigenvalues of $ \{ J_2 , \cdots , J_n \} $, which are easily read off from standard tableaux, i.e. Young diagrams with $n$ boxes, with the numbers $\{ 1, 2, \cdots , n \}$ entered into the boxes according to specified constraints. This spectral information gives precisely the labels for the nodes of  the degeneracy graphs we have described. Thus the invertibility of the matrix of coefficients 
\eqref{MatCoeffs} and the resulting construction of projectors has immediate applications to the Jucys-Murphy algebras.

The primitive idempotents of the Jucys–Murphy algebra coincide with the diagonal matrix units in Young’s seminormal representation \cite{JamesKerber,Jucys,Murphy,OkouVersh}. The JM generators therefore provide a complete set of commuting observables with simple joint spectrum, whose eigenlines correspond to standard tableaux and whose spectral projectors are precisely the seminormal rank-one projectors. Our construction recovers these projectors directly from their joint eigenvalues—encoded by the degeneracy graph—without invoking the explicit tableau-based action of permutations.

\subsubsection{Jucys-Murphy elements for $S_3$}

Below we use the  layered  branching graph whose vertices are the Symmetric group standard Young tableaux. Level $k$ of the graph consists of all standard tableaux for Young diagrams with $k$ boxes. Thus, the total number of vertices in layer $k$ is $\sum_{R} d_{R}$, where $d_{R}$ is the dimension of irrep $R$ of $S_k$. A vertex in layer $k$ is connected to a vertex in layer $k+1$ by an edge if the standard tableau in $k+1$ can be obtained from the tableau in $k$ by adding a box labelled by $k+1$. We begin from $k=1$ and go up to $k=3$. 

The eigenvalue of a Jucys-Murphy element $J_k$, acting on a given tableau, is the content of the box  containing the label $k$. The content for the box at row $i$ and column $j$ is defined as   $ ( j -i )$. 
 The contents of the $S_3$ irreps are as follows: 
\bea
\ytableausetup{boxsize=1.25em}
\begin{ytableau}
       \none &0&1 & 2
\end{ytableau}, \hspace{10pt} \ytableausetup{boxsize=1.25em}
\begin{ytableau}
       \none &0&1 \\
       	\none &{\scriptstyle -1}\\
\end{ytableau}, \hspace{10pt} \ytableausetup{boxsize=1.25em}
\begin{ytableau}
       \none &0 \\
       	\none &{\scriptstyle -1}\\
	\none &{\scriptstyle -2}\\
\end{ytableau}
\eea
On the far right-hand side of the branching graph we give the eigenvalues for $(J_2 , J_3)$ on each of the $S_3$ tableaux.
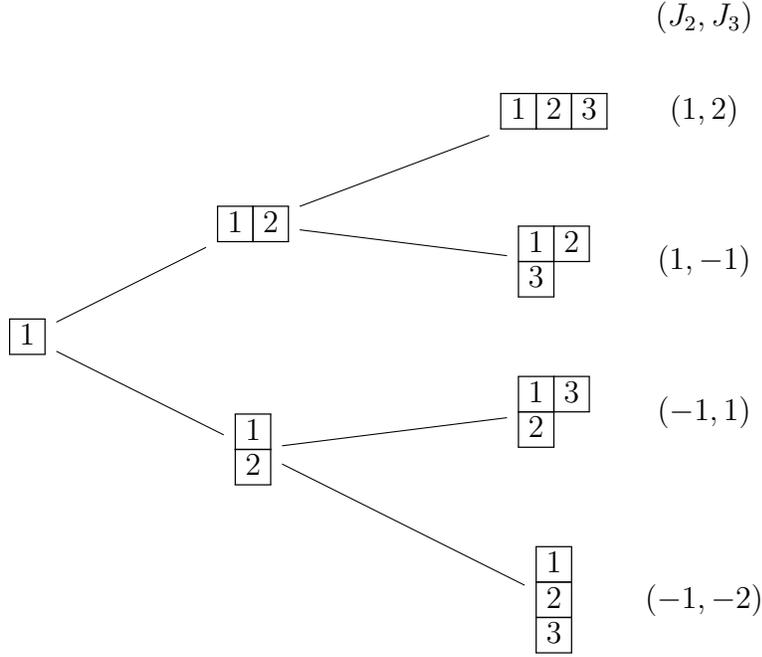
\vspace{5pt}

\begin{centering}
\begin{tikzpicture}[
  every node/.style={anchor=center}
]

\ytableausetup{boxsize=1.1em}

\node (nJ2J3) at (9,4.25) {$(J_2 , J_3 )$};

% Root
\node (n1) at (0,0) {$
\begin{ytableau}
1
\end{ytableau}
$};

% Level 2
\node (n12) at (3,1.5) {$
\begin{ytableau}
1 & 2
\end{ytableau}
$};

\node (n21) at (3,-1.5) {$
\begin{ytableau}
1 \\
2
\end{ytableau}
$};

% Level 3 from top branch
\node (n123) at (7,3) {$
\begin{ytableau}
1 & 2 & 3
\end{ytableau}
$};

\node (n123eig) at (9,3) {$(1,2)$};

\node (n12_3) at (7,1) {$
\begin{ytableau}
1 & 2 \\
3
\end{ytableau}
$};

\node (n123eig2) at (9,1) {$(1,-1)$};

% Level 3 from bottom branch
\node (n13_2) at (7,-1) {$
\begin{ytableau}
1 & 3 \\
2
\end{ytableau}
$};

\node (n123eig3) at (9,-1) {$(-1 , 1)$};

\node (n1_23) at (7,-3.5) {$
\begin{ytableau}
1 \\
2 \\
3
\end{ytableau}
$};

\node (n123eig4) at (9,-3.5) {$(-1,-2)$};

% Edges
\draw (n1) -- (n12);
\draw (n1) -- (n21);

\draw (n12) -- (n123);
\draw (n12) -- (n12_3);

\draw (n21) -- (n13_2);
\draw (n21) -- (n1_23);

\end{tikzpicture}
\end{centering}

The data specifying the corresponding degeneracy graph is
\bea
	\cD &=& (2;2,4,p_1)\\ 
	p_1 &=& (2,2).
\eea
The degeneracy graph is shown below. 

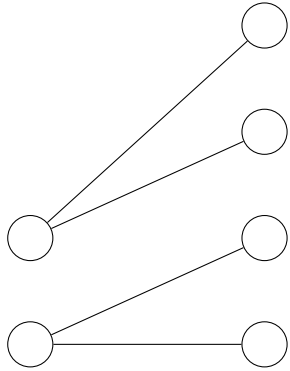
\begin{center}
\begin{tikzpicture}[
  node distance = 8mm and 25mm,
  every node/.style = {circle, draw, minimum size=6mm},
]

%------------------------------------------------
% Left column: x-nodes
%------------------------------------------------
\node (x1) {};
\node (x2) [above=of x1] {};

%------------------------------------------------
% Right column: y-nodes
%------------------------------------------------
\node (y11) [right=of x1] {};
\node (y12) [above=of y11] {};
\node (y21) [above=of y12] {};
\node (y22) [above=of y21] {};

%------------------------------------------------
% Edges
%------------------------------------------------
\draw (x1) -- (y11);
\draw (x1) -- (y12);
\draw (x2) -- (y21);
\draw (x2) -- (y22);

\end{tikzpicture}
\end{center}

Furthermore, 
\bea
	&&\cS^{(1)}_{\left[1\right]} = \{ 1,2 \} , \hspace{10pt} \cS^{(1)}_{\left[2\right]} = \{ 1,2 \} \\
	&&\left|\cS^{(1)}_{\left[1\right]}\right| = 2, \hspace{10pt}  \left|\cS^{(1)}_{\left[2\right]}\right| = 2.
\eea
Thus, the monomials are
\bea
	\cS(\cA_1 \rightarrow \cA_2) = \{ 1 , J_2  \} \sqcup  \{ 1 , J_2 \}\times J_3.
\eea
The eigenvalue equation for the $J_{k}$ acting on the projector basis is
\bea
\label{eq:JMeigenvalueEqns}
	J_k P_I = c_I(k) P_I,
\eea
where $I$ ranges over the tableaux in the final layer in the graph, and the eigenvalue $c_I(k)$ is the content of the box labeled $k$ in the $I$th tableau. The expansion of the Jucys-Murphy elements in terms of the projector basis is
\bea
	(1 , J_2 , J_3 , J_2 J_3) = (P_1 , P_2 , P_3 , P_4) \left(\begin{array}{cccc}1 & 1 & 2 & 2 \\1 & 1 & -1 & -1 \\1 & -1 & 1 & -1 \\1 & -1 & -2 & 2\end{array}\right),
\eea
and the inverse of the above system is
\bea
	(P_1 , P_2 , P_3 , P_4) = (1 , J_2 , J_3 , J_2 J_3) \, \frac{1}{6} \left(\begin{array}{cccc}1 & 2 & 2 & 1 \\1 & 2 & -2 & -1 \\1 & -1 & 1 & -1 \\1 & -1 & -1 & 1\end{array}\right).
\eea
Using these relations, we can easily verify the eigenvalue equations (\ref{eq:JMeigenvalueEqns}). For example,
\bea
	J_2 P_1 = \frac{1}{6}\left( J_2 + J^{2}_{2} + J_2 J_3 + J^{2}_{2}J_3 \right).
\eea
We have $J_2 = (12)$, which means $J^{2}_{2} = 1$. Thus,
\bea
	J_2 P_1 = \frac{1}{6}\left( J_2 + 1 + J_2 J_3 + J_3 \right) = 1 P_1.
\eea
We can finally express $J_2$ and $J_3$ in terms of the projector basis
\bea
	J_2 = \left(\begin{array}{cccc}1 & 0 & 0 & 0 \\0 & 1 & 0 & 0 \\0 & 0 & -1 & 0 \\0 & 0 & 0 & -1\end{array}\right) , \hspace{20pt} J_3 = \left(\begin{array}{cccc}2 & 0 & 0 & 0 \\0 & -1 & 0 & 0 \\0 & 0 & 1 & 0 \\0 & 0 & 0 & -2\end{array}\right).
\eea

\subsubsection{Jucys-Murphy elements for $S_4$}
 
The branching graph from layer $k = 1$ up to $k=4$ is given below with the eigenvalues of $J_2 , J_3 $ and $J_4$ given on the far right hand side. 

\begin{tikzpicture}[
  every node/.style={anchor=center}
]

\node (J2J3J4) at (14,9) {$(J_2,J_3,J_4)$};

\ytableausetup{boxsize=1.1em}

%-----------------------
% Level 1
%-----------------------
\node (n1) at (0,0) {$
\begin{ytableau}
1
\end{ytableau}
$};

%-----------------------
% Level 2
%-----------------------
\node (n12) at (3,1.5) {$
\begin{ytableau}
1 & 2
\end{ytableau}
$};

\node (n21) at (3,-1.5) {$
\begin{ytableau}
1 \\
2
\end{ytableau}
$};

%-----------------------
% Level 3
%-----------------------
\node (n123) at (7,3) {$
\begin{ytableau}
1 & 2 & 3
\end{ytableau}
$};

\node (n12_3) at (7,1) {$
\begin{ytableau}
1 & 2 \\
3
\end{ytableau}
$};

\node (n13_2) at (7,-1) {$
\begin{ytableau}
1 & 3 \\
2
\end{ytableau}
$};

\node (n1_23) at (7,-3.5) {$
\begin{ytableau}
1 \\
2 \\
3
\end{ytableau}
$};

%-----------------------
% Level 4 (NEW)
%-----------------------

% From 123
\node (n1234) at (11,8.0) {$
\begin{ytableau}
1 & 2 & 3 & 4
\end{ytableau}
$};

\node (Tab1) at (14,8) {$(1,2,3)$};

\node (n123_4) at (11,6.3) {$
\begin{ytableau}
1 & 2 & 3 \\
4
\end{ytableau}
$};

\node (Tab2) at (14,6.3) {$(1,2,-1)$};

% From 12/3
\node (n12_34) at (11,5) {$
\begin{ytableau}
1 & 2 & 4 \\
3
\end{ytableau}
$};

\node (Tab3) at (14,5) {$(1,-1,2)$};

\node (n12_3_4) at (11,3.25) {$
\begin{ytableau}
1 & 2 \\
3 & 4
\end{ytableau}
$};

\node (Tab4) at (14,3.25) {$(1,-1,0)$};

\node (n12_3__4) at (11,1.5) {$
\begin{ytableau}
1 & 2 \\
3 \\
4
\end{ytableau}
$};

\node (Tab5) at (14,1.5) {$(1,-1,-2)$};

\node (n13_24) at (11,-0.5) {$
\begin{ytableau}
1 & 3 & 4 \\
2
\end{ytableau}
$};

\node (Tab6) at (14,-0.5) {$(-1,1,2)$};

\node (n13_2_4) at (11,-1.8) {$
\begin{ytableau}
1 & 3 \\
2 & 4
\end{ytableau}
$};

\node (Tab7) at (14,-1.8) {$(-1,1,0)$};

\node (n13_2__4) at (11,-4) {$
\begin{ytableau}
1 & 3 \\
2 \\
4
\end{ytableau}
$};

\node (Tab8) at (14,-4) {$(-1,1,-2)$};

% From 1/2/3
\node (n1_23_4) at (11,-6.0) {$
\begin{ytableau}
1 & 4 \\
2 \\
3
\end{ytableau}
$};

\node (Tab9) at (14,-6) {$(-1,-2,1)$};

\node (n1_234) at (11,-8.5) {$
\begin{ytableau}
1 \\
2 \\
3 \\
4
\end{ytableau}
$};

\node (Tab10) at (14,-8.5) {$(-1,-2,-3)$};

%-----------------------
% Edges (original)
%-----------------------
\draw (n1) -- (n12);
\draw (n1) -- (n21);

\draw (n12) -- (n123);
\draw (n12) -- (n12_3);

\draw (n21) -- (n13_2);
\draw (n21) -- (n1_23);

%-----------------------
% Edges to level 4
%-----------------------

% From 123
\draw (n123) -- (n1234);
\draw (n123) -- (n123_4);

% From 12/3
\draw (n12_3) -- (n12_34);
\draw (n12_3) -- (n12_3_4);
\draw (n12_3) -- (n12_3__4);

% From 13/2
\draw (n13_2) -- (n13_24);
\draw (n13_2) -- (n13_2_4);
\draw (n13_2) -- (n13_2__4);

% From 1/2/3
\draw (n1_23) -- (n1_23_4);
\draw (n1_23) -- (n1_234);
\end{tikzpicture}

The data for the corresponding degeneracy graph is given by
\bea
	\cD &=& (3 ; 2,4,10 ; p_1 , c_2)\\
	p_1 &=& (2,2) , \hspace{10pt} c_2 = (2,3,3,2) .
\eea

The graph is shown below

\begin{center}
\begin{tikzpicture}[
  node distance = 6mm and 25mm,
  every node/.style = {circle, draw, minimum size=6mm},
]

%------------------------------------------------
% Left column: x-nodes
%------------------------------------------------
\node (x1) {};
\node (x2) [above=of x1] {};

%------------------------------------------------
% Middle column: y-nodes
%------------------------------------------------
\node (y11) [right=of x1] {};
\node (y12) [above=of y11] {};
\node (y21) [above=of y12] {};
\node (y22) [above=of y21] {};

%------------------------------------------------
% Right column: z-nodes (new layer)
%------------------------------------------------
% children of y11 (2)
\node (z111) [right=of y11, yshift=-50mm] {};
\node (z112) [right=of y11, yshift=-30mm] {};

% children of y12 (3)
\node (z121) [right=of y12, yshift=-25mm] {};
\node (z122) [right=of y12, yshift=-9mm] {};
\node (z123) [right=of y12, yshift=5mm] {};

% children of y21 (3)
\node (z211) [right=of y21, yshift=5mm] {};
\node (z212) [right=of y21, yshift=20mm] {};
\node (z213) [right=of y21, yshift=35mm] {};

% children of y22 (2)
\node (z221) [right=of y22, yshift=35mm] {};
\node (z222) [right=of y22, yshift=50mm] {};

%------------------------------------------------
% Edges: x -> y
%------------------------------------------------
\draw (x1) -- (y11);
\draw (x1) -- (y12);
\draw (x2) -- (y21);
\draw (x2) -- (y22);

%------------------------------------------------
% Edges: y -> z (new layer)
%------------------------------------------------
\draw (y11) -- (z111);
\draw (y11) -- (z112);

\draw (y12) -- (z121);
\draw (y12) -- (z122);
\draw (y12) -- (z123);

\draw (y21) -- (z211);
\draw (y21) -- (z212);
\draw (y21) -- (z213);

\draw (y22) -- (z221);
\draw (y22) -- (z222);

\end{tikzpicture}
\end{center}

Furthermore,
\bea
	&&\left| \cS^{(1)}_{\left[1,1\right]} \right| = 2, \hspace{10pt} \left| \cS^{(1)}_{\left[2,1\right]} \right| = 2, \hspace{10pt} \left| \cS^{(1)}_{\left[1,2\right]} \right| = 2 \\
	&& \left| \cS^{(1)}_{\left[2,2\right]} \right| = 2, \hspace{10pt} \left| \cS^{(1)}_{\left[1,3\right]} \right| = 2.
\eea
The monomials are thus given by
\bea
\label{eq:MonosJMS4}
	&&\cS(\cA_1 \rightarrow \cA_2 \rightarrow \cA_3) =\\
	&& \Big( \{ 1 , J_2 \} \Big) \sqcup \Big( \{ 1 , J_2 \} \times J_3 \Big) \sqcup \Big( \{ 1 , J_2 \} \times J_4 \Big) \sqcup \Big( \{ 1 , J_2 \} \times J_3 \times J_4 \Big) \sqcup \Big( \{ 1 , J_2 \} \times J^2_4 \Big).\nonumber
\eea
We can write the monomials above in terms of the projector basis using the eigenvalues
\bea
	(1,J_2,J_3,J_4,J_2 J_3 , J_2 J_4 , J_3,J_4 , J^{2}_{4} , J_2 J_3 J_4 ,J_2 J^{2}_{4} ) =  (P_1 , P_2 , \cdots , P_{10}) \cM
\eea
where the matrix of coefficients $\cM$ is given by
\bea
	 \cM =  \left(
\begin{array}{cccccccccc}
 1 & 1 & 2 & 3 & 2 & 3 & 6 & 9 & 6 & 9 \\
 1 & 1 & 2 & -1 & 2 & -1 & -2 & 1 & -2 & 1 \\
 1 & 1 & -1 & 2 & -1 & 2 & -2 & 4 & -2 & 4 \\
 1 & 1 & -1 & 0 & -1 & 0 & 0 & 0 & 0 & 0 \\
 1 & 1 & -1 & -2 & -1 & -2 & 2 & 4 & 2 & 4 \\
 1 & -1 & 1 & 2 & -1 & -2 & 2 & 4 & -2 & -4 \\
 1 & -1 & 1 & 0 & -1 & 0 & 0 & 0 & 0 & 0 \\
 1 & -1 & 1 & -2 & -1 & 2 & -2 & 4 & 2 & -4 \\
 1 & -1 & -2 & 1 & 2 & -1 & -2 & 1 & 2 & -1 \\
 1 & -1 & -2 & -3 & 2 & 3 & 6 & 9 & -6 & -9 \\
\end{array}
\right).
\eea
The inverse of $\cM$ is
\bea
	\cM^{-1} = \left(
\begin{array}{cccccccccc}
 \frac{1}{24} & \frac{1}{8} & -\frac{1}{16} & \frac{11}{24} & -\frac{1}{16} &
   -\frac{1}{16} & \frac{11}{24} & -\frac{1}{16} & \frac{1}{8} & \frac{1}{24} \\
 \frac{1}{24} & \frac{1}{8} & -\frac{1}{16} & \frac{11}{24} & -\frac{1}{16} &
   \frac{1}{16} & -\frac{11}{24} & \frac{1}{16} & -\frac{1}{8} & -\frac{1}{24} \\
 \frac{1}{24} & \frac{1}{8} & -\frac{1}{16} & -\frac{1}{24} & -\frac{1}{16} &
   \frac{1}{16} & \frac{1}{24} & \frac{1}{16} & -\frac{1}{8} & -\frac{1}{24} \\
 \frac{1}{24} & -\frac{1}{24} & \frac{1}{24} & \frac{1}{12} & -\frac{1}{8} & \frac{1}{8}
   & -\frac{1}{12} & -\frac{1}{24} & \frac{1}{24} & -\frac{1}{24} \\
 \frac{1}{24} & \frac{1}{8} & -\frac{1}{16} & -\frac{1}{24} & -\frac{1}{16} &
   -\frac{1}{16} & -\frac{1}{24} & -\frac{1}{16} & \frac{1}{8} & \frac{1}{24} \\
 \frac{1}{24} & -\frac{1}{24} & \frac{1}{24} & \frac{1}{12} & -\frac{1}{8} & -\frac{1}{8}
   & \frac{1}{12} & \frac{1}{24} & -\frac{1}{24} & \frac{1}{24} \\
 \frac{1}{24} & -\frac{1}{24} & -\frac{1}{12} & \frac{1}{12} & 0 & 0 & \frac{1}{12} &
   -\frac{1}{12} & -\frac{1}{24} & \frac{1}{24} \\
 0 & 0 & \frac{1}{16} & -\frac{1}{8} & \frac{1}{16} & \frac{1}{16} & -\frac{1}{8} &
   \frac{1}{16} & 0 & 0 \\
 \frac{1}{24} & -\frac{1}{24} & -\frac{1}{12} & \frac{1}{12} & 0 & 0 & -\frac{1}{12} &
   \frac{1}{12} & \frac{1}{24} & -\frac{1}{24} \\
 0 & 0 & \frac{1}{16} & -\frac{1}{8} & \frac{1}{16} & -\frac{1}{16} & \frac{1}{8} &
   -\frac{1}{16} & 0 & 0 \\
\end{array}
\right).
\eea
Thus, using $\cM^{-1}$ we can write the projector basis in terms of the monomial basis (\ref{eq:MonosJMS4}). For example, $P_{2}$, which corresponds to the tableau 
\bea \label{eq:EigValEqP2} \ytableausetup{boxsize=1.0em}
\begin{ytableau}
1 & 2 & 3 \\
4
\end{ytableau} , \hspace{10pt} \hbox{with contents}  \hspace{10pt} \ytableausetup{boxsize=1.0em}
\begin{ytableau}
0 & 1 & 2 \\
{\scriptstyle -1}
\end{ytableau} , \eea
is obtained from $\cM^{-1}$ and given by
\bea
	P_{2} = \frac{1}{8} +\frac{J_2}{8} +\frac{J_3}{8} -\frac{J_4}{24} + \frac{J_2 J_3}{8 } -\frac{J_2 J_4}{24} -\frac{J_3 J_4}{24} - \frac{J_2 J_3 J_4}{24}   .
\eea
Next, we verify the eigenvalue equation for $J_2  J_3  J_4 $ on $P_2$. From (\ref{eq:EigValEqP2}), the eigenvalue of $J_2 J_3 J_4 = -2$. 
\bea
	J_2  J_3  J_4 P_2 &=& J_2  J_3  J_4 \Big( \frac{1}{8} +\frac{J_2}{8} +\frac{J_3}{8} -\frac{J_4}{24} + \frac{J_2 J_3}{8 } -\frac{J_2 J_4}{24} -\frac{J_3 J_4}{24} - \frac{1}{24} J_2 J_3 J_4 \Big)\nonumber\\
	&=& \frac{1}{8}J_2  J_3  J_4  +\frac{1}{8}  J_3  J_4  +\frac{1}{8}J_2  J^2_3  J_4  -\frac{1}{24}J_2  J_3  J^2_4  + \frac{ J^2_3 J_4}{8 } -\frac{J_3  J^2_4  }{24} -\frac{J_2  J^2_3  J^2_4 }{24} - \frac{1}{24}  J^2_3 J^2_4, \nonumber
 \eea
 where we have used $J^{2}_{2} = 1$. Expanding the above products in terms of the basis monomials gives
 \bea
 	J_2  J_3  J_4 P_2 &=& \frac{1}{8}J_2  J_3  J_4  +\frac{1}{8}  J_3  J_4 + \frac{1}{8}\big( 2 J_2 J_4+J_3 J_4 \big) - \frac{1}{24}\big( -J_4^2+2 J_2 J_4+2 J_3 J_4+3 J_2 J_3+3 \big)\nonumber\\
	&& + \frac{1}{8}\big( J_2 J_3 J_4+2 J_4 \big) - \frac{1}{24} \big( -J_2 J_4^2+2 J_2 J_3 J_4+2 J_4+3 J_2+3 J_3 \big) \nonumber\\
	&& - \frac{1}{24} \big( J_2 J_4^2+2 J_2 J_3 J_4+2 J_4+3 J_2+3 J_3 \big)  - \frac{1}{24} \big( J_4^2+2 J_2 J_4+2 J_3 J_4+3 J_2 J_3+3 \big) \nonumber
 \eea
Finally, collecting the terms gives
\bea
	J_2  J_3  J_4 P_2 &=&  -\frac{1}{4} - \frac{J_2}{4} - \frac{J_3}{4} + \frac{J_4}{12} - \frac{1}{4} J_2 J_3 + \frac{J_2 J_4}{12}  +\frac{J_3 J_4}{12}  + \frac{1}{12} J_2 J_3 J_4 , \nonumber\\
		&=& -2 \bigg(  \frac{1}{8} +\frac{J_2}{8} +\frac{J_3}{8} -\frac{J_4}{24} + \frac{J_2 J_3}{8 } -\frac{J_2 J_4}{24} -\frac{J_3 J_4}{24} - \frac{1}{24} J_2 J_3 J_4 \bigg)\nonumber\\
		&=& -2 P_2.
\eea

\subsection{ Application to matrix units of semi-simple algebras and multi-matrix invariants } 
\label{sec:MUandMultMat} 

The Jucys-Murphy elements considered in section  \ref{sec:JM} generate maximally commuting sub-algebras 
of $ \mC ( S_n)$. The dimension of the algebra they generate is equal to the sum of dimensions of the irreducible representations of 
\bea 
\sum_{ R \vdash n  } d_R 
\eea
There is a basis of matrix units of $ \mC ( S_n)$ given by 
\bea\label{QsandDs}  
Q^R_{ ij} = { d_R \over n!  } \sum_{ \sigma \in S_n } D^R_{ ji } ( \sigma^{-1}  ) \sigma  
\eea
The terminology "matrix units" refers to the fact that these elements of the group algebra multiply as elementary matrices in distinct blocks labelled by $R$ 
\bea 
Q^R_{ij} Q^S_{ kl} = \delta_{ j k } Q^R_{ il} \delta^{ RS} 
\eea
These form a basis for $ \mC ( S_n) $ and the number of these is $ \sum_{ R } d_R^2$, which is equal to the order of the group. The projectors in the algebra generated by Jucys-Murphy elements can be identified with the diagonal matrix units, which can also be constructed using the  Young-tableaux based method for constructing the matrix elements $D^R_{ij} ( \sigma )$ in the semi-normal representation.

The Wedderburn-Artin theorem (see e.g. \cite{CurtRein}) ensures that any associative semi-simple algebra admits a basis of matrix units. An analog of the formula in terms of matrix elements of irreducible  representations \cite{RamNotes} is 
\bea\label{QsandDsGenA}  
Q^R_{ ij} = { t_{ R } }  \sum_{ b  \in \cB ( \cA )  } D^R_{ ji } ( b^*   ) b 
\eea
$\cB ( \cA ) $ is a basis of the algebra $ \cA$,$t_R$ is a constant which can be determined by using the matrix unit property,  $b^*$ is the dual of $b$ under the non-degenerate bilinear pairing which exists due to the semi-simplicity property. 

For more general algebras beyond the well-studied symmetric group, reasonably practical algorithms for finding the matrix units tend to be rare and there are some interesting recent efforts this gap, see e.g. \cite{BPR,AdrianThesis,Campbell}.

A number of associative algebras of interest in the construction of orthogonal bases of 
multi-matrix invariants \cite{KR,BHR,BCD,BCS,EHS} or orthogonal bases of tensor invariants \cite{BGSR,DiazSJ} were identified in \cite{PCA,BGSR2}. Examples are 
\begin{itemize} 

\item $ \cA( m , n ) $ relevant for 2-matrix invariants: the sub-algebra of the group algebra of the symmetric group $S_{ m+n} $ of permutations of $ [m+n] =  \{ 1, 2, \cdots , m+n \}$, which commutes with permutations in the subgroup $ S_m \times S_n$, consisting of permutations which preserve the subsets 
$ \{ 1, \cdots , m \} \subset [ m+n] $ and $ \{ m+1 , \cdots , m+n \} \subset [ m+n ] $. 

\item $  \cA ( \cB_{ N } ( m  , n ) )$ also relevant to 2-matrix invariants : the sub-algebra of the walled Brauer algebra $B_N(m,n)$, which commutes with permutations in a sub-algebra  $  \mC ( S_m \times S_n ) \subset B_N ( m,n)$. 

\item $ \cK(n) $ relevant for complex 3-index tensor invariants invariant under $ U(N) \times U(N) \times U(N)$ : the sub-algebra of $ \mC ( S_n \times S_n \times S_n) $ which commutes with the diagonally embedded $ \mC ( S_n )$. 

\end{itemize} 

The algebra  $  \cA ( \cB_{ N } ( m  , n ) )$  is also useful in the optimisation of algorithms 
in port-based quantum teleportation \cite{HSM,MHM}, and closely related symmetry-adapted representation-theoretic structures which  underlie  developments of the protocol \cite{CMMSWW,GBO}.
For any of these algebras, let us assume we have found a set of generators of a maximal commutative sub-algebra analogous to the Jucys-Murphy elements for $ \mC ( S_n)$. Applying the monomial basis algorithm for projectors as we illustrated in section \ref{sec:JM}, we can produce the $Q^R_{ii}$ of diagonal matrix units. The off-diagonal matrix units  $Q^R_{ij} $ can be constructed by using 
eigenvalue equations in the algebra.   $Q^R_{ij} $  an eigenstate of $Q^{R}_{ii}$ with unit eigenvalue under left multiplication, and an eigenstate of  $Q^R_{jj} $ with unit eigenvalue under right multiplication.

Finally, it is worth noting that in certain physically relevant examples the underlying algebra becomes non-semisimple in specific parameter regimes. A prominent example is the walled Brauer algebra \( B_N(m,n) \), which fails to be semisimple when \( m+n > N \). In such cases, combinatorial tools—such as restricted Bratteli diagrams—can be used to transport representation-theoretic information from the semisimple regime into the non-semisimple setting \cite{RamStud}. The monomial basis construction for maximal commutative subalgebras developed here in the semisimple case may therefore provide useful structural input for analysing the non-semisimple regime as well.

\section{ Summary and Outlook }

We considered sequences of sub-algebras of a commutative associative semi-simple algebra $ \cA = \cA_L$, 
\bea 
\cA_1 \rightarrow \cA_2 \rightarrow \cdots \rightarrow \cA_L 
\eea
 determined by choosing an ordered set of non-linearly generating elements $ \{ \cC_1 , \cC_2 , \cdots , \cC_L \} $ for the algebra. Degeneracy graphs were defined using the eigenvalues of 
 the generators. 
  The main conjecture \eqref{mainconj} gives a basis for $ \cA_L $ in terms of 
 monomials in the generators. The determinant of the matrix  of coefficients relating the monomial basis to the projector basis is conjectured to be of the form given in \eqref{FactoredDet} and \eqref{DetConj2}. Verifying the determinant conjectures guarantees that the matrix of coefficients is invertible and the monomials specified  by \eqref{mainconj} indeed give a basis for $ \cA_L$. Verifications for degeneracy graphs with varying numbers of nodes and with number of layers $L$ up to $5$ have been done, and the computer  code written in SAGE is provided. 

We have illustrated  applications of these results to the construction of projectors in centres of symmetric group algebras as well as maximally commuting sub-algebras of symmetric group sub-algebras generated by Jucys-Murphy elements. We outlined an important future research direction of applying the construction of projectors to maximally commuting sub-algebras of
the permutation centraliser algebras which arise in the study of orthogonal bases of multi-matrix invariants and in quantum theory. This can be combined with eigenvalue methods to give a construction for the full set of matrix units of the algebras. This will provide valuable results for the study of correlators of multi-matrix models, the AdS/CFT correspondence and quantum information theory. 

We observed the important {\it order ideal } property of the proposed monomial  basis $ \cS ( \cA_1 , \cdots , \cA_L ) $. This suggests that the algebras $ \cA_L$  under study can be further illuminated by regarding them  as quotients by an ideal $I$ of the  polynomial ring generated $ \{ \cC_1 , \cC_2 , \cdots , \cC_L \} $. The ideal is spanned as a vector space over $\mC$ by powers of the generators which are complementary to those in the basis. 

The fundamental motivating example which led to the present investigation is the fact that 
the cycle operators $ \{ T_2 , T_3 , \cdots , T_{ k_*(n) } \} $, with $k_*(n)$ growing slowly with $n$, provide a non-linear generating set for the centres of symmetric group algebras $ \mC ( S_n) $. As discussed earlier, $k_*(n) $ is also the number of normalised characters needed to distinguish all Young diagrams with $n$ boxes.   A closely related sequence, of physical interest in AdS/CFT, is $ k_* ( n , N ) $ which is the  number of normalised characters needed to distinguish all Young diagrams with $n$ boxes and no more than $N$ rows. The unexpectedly rich structure of monomial bases and graph determinants, along with the relevance to  quantum information theory and AdS/CFT, suggest that further investigations of $ k_* (n)$ and $ k_*  ( n , N )  $ in limits of large $n,N$ will be fruitful from physical as well as mathematical perspectives.

\begin{center} 
\section*{Acknowledgments}
\end{center} 

It is a pleasure to acknowledge useful conversations related to the subject of this paper with Mahesh Balasubramanian, Matt Buican, Robert de Mello Koch, Brian Dolan, Thomas Fink,  Yang-Hui He,  Chris Hull, Vishnu Jejjala, Yang Lei, Charles Nash,  Denjoe O' Connor, Adrian Padellaro, Micha\l $~$ Studzi\'nski, Ryo Suzuki. SR is supported by the Science and Technology Facilities Council (STFC) Consolidated Grant ST/X00063X/1  
``Amplitudes, strings and duality''. SR was supported by a Visiting Professorship at Dublin Institute for Advanced Studies,  held during 2024, and  is supported by ongoing Royal Society International Exchanges grant held jointly with Yang Lei. We also grateful to the London Institute of Mathematical Sciences for providing a stimulating mathematical environment for  discussions with several colleagues.  The authors acknowledge the assistance of ChatGPT in developing the computational code for the degeneracy graphs, based on the authors’ combinatorial formulation, which facilitated the efficient completion of this work.

\vskip.5cm 

\begin{appendix}

\section{ Examples and  guide to the SAGE code }
\label{sec:AppCode} 

The arXiv upload of this paper  is accompanied by an  ancillary upload consisting of three sage files:  Degeneracy-Graph-Construction.ipynb, DetConjecture1.ipynb and DetConjecture2.ipynb.

In the Degeneracy-Graph-Construction.ipynb, the code for constructing the degeneracy graphs for a general number of layers $L$, with specified $ D_1 < D_2 < \cdots < D_L$, and partition $p_1$ and compositions $ \{ c_1 , c_2 , \cdots , c_{ L-1} \} $ is given. 

As an example, the following commands specify the data $ D_1 = 2 , D_2 =4$, the partition $ p_1 = [ 2,2] $. This is a 2-layer graph ($L=2$), where the list of compositions is empty. 
\begin{verbatim} 
D = [2, 4]
p1 = [2, 2]
comps = { } 
print ( mons_L(D, p1, comps)  )   
print ( mat_L(D, p1, comps)   )  
print ( det_L(D, p1, comps) )     
graph_L(D, p1, comps)  
\end{verbatim} 
The first output is the list of monomials
\bea\label{monomExamp1} 
\{ 1,\;\cC_{1},\;\cC_{2},\;\cC_{1}\cC_{2} \} 
\eea 
The second output is the matrix for the change of basis from the monomials to the projectors. 
The coefficients along the columns are the expansion coefficients for the monomials in the list above in terms of the projectors 
\[
M=\begin{pmatrix}
1 & x^{(1)}_{1} & x^{(2)}_{1} & x^{(1)}_{1}x^{(2)}_{1}\\
1 & x^{(1)}_{1} & x^{(2)}_{2} & x^{(1)}_{1}x^{(2)}_{2}\\
1 & x^{(1)}_{2} & x^{(2)}_{3} & x^{(1)}_{2}x^{(2)}_{3}\\
1 & x^{(1)}_{2} & x^{(2)}_{4} & x^{(1)}_{2}x^{(2)}_{4}
\end{pmatrix}.
\]
The determinant takes the factorised form 
\bea\label{DetExamp1} 
\det M \;=\; -\bigl(x^{(1)}_{1}-x^{(1)}_{2}\bigr)^{2}\,\bigl(x^{(2)}_{1}-x^{(2)}_{2}\bigr)\,\bigl(x^{(2)}_{3}-x^{(2)}_{4}\bigr).
\eea 
The graph corresponding to the specified data is
% Requires: \usepackage{tikz}
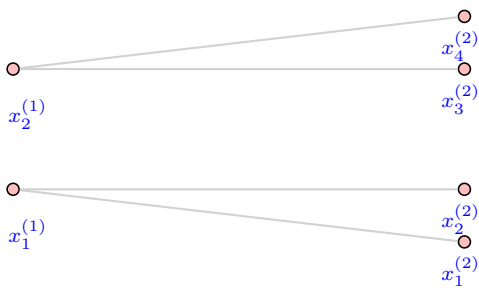
\[ 
\begin{tikzpicture}[x=1cm,y=1cm]

\tikzset{
  pt/.style={circle, draw=black, fill=red!25, line width=0.6pt, inner sep=1.6pt},
  ed/.style={draw=gray!35, line width=0.8pt},
  lab/.style={text=blue, font=\scriptsize}
}

% --- coordinates (gently compressed vertical spacing) ---
\node[pt] (x1_2) at (0,  1.0) {};
\node[pt] (x1_1) at (0, -0.6) {};

\node[pt] (x2_4) at (6,  1.7) {};
\node[pt] (x2_3) at (6,  1.0) {};
\node[pt] (x2_2) at (6, -0.6) {};
\node[pt] (x2_1) at (6, -1.3) {};

% --- edges ---
\draw[ed] (x1_2) -- (x2_3);
\draw[ed] (x1_2) -- (x2_4);
\draw[ed] (x1_1) -- (x2_2);
\draw[ed] (x1_1) -- (x2_1);

% --- labels ---
\node[lab, anchor=north east] at ($(x1_2)+(0.60,-0.25)$) {$x^{(1)}_{2}$};
\node[lab, anchor=north east] at ($(x1_1)+(0.60,-0.25)$) {$x^{(1)}_{1}$};

\node[lab, anchor=north west] at ($(x2_4)+(-0.45,-0.05)$) {$x^{(2)}_{4}$};
\node[lab, anchor=north west] at ($(x2_3)+(-0.45,-0.05)$) {$x^{(2)}_{3}$};
\node[lab, anchor=north west] at ($(x2_2)+(-0.45,-0.05)$) {$x^{(2)}_{2}$};
\node[lab, anchor=north west] at ($(x2_1)+(-0.45,-0.05)$) {$x^{(2)}_{1}$};

\end{tikzpicture}
\] 
The determinant is non-vanishing for $ x^{(1)}_1 \ne x^{(1)}_2$ and for $ x^{(2)}_a \ne x^{(2)}_b $, whenever the nodes corresponding to $ x^{(2)}_a $ and $x^{(2)}_b$ share the same parent. 

\vskip.3cm

The next example is an $L=3$ graph, with $ D_1  = 2, D_2 = 4, D_3 = 7   $. The partition 
$ p = [ 2,2]$ specifies that the two nodes in the first layer each  connect to two nodes in the second layer. The composition $c_2 = [ 2,1,2,2] $ specify that connections between the second layer and the third layer, starting from the lowest node in the second layer.  This data is specified in the sagemath code as follows, followed by an instruction to print the data. 
\begin{verbatim} 
D = [2, 4, 7]
p1 = [2, 2]
comps = {2: [2, 1, 2, 2]}  
print( graph_L(D, p1, comps))
print ( mons_L(D, p1, comps)  )  
print ( mat_L(D, p1, comps)   )   
det_L(D, p1, comps) 
\end{verbatim} 
The graph is 
% Requires: \usepackage{tikz}
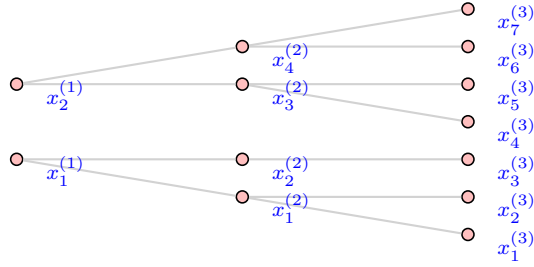
\begin{center}
\begin{tikzpicture}[x=1cm,y=1cm]

\tikzset{
  pt/.style={circle, draw=black, fill=red!25, line width=0.6pt, inner sep=1.6pt},
  ed/.style={draw=gray!35, line width=0.8pt},
  lab/.style={text=blue, font=\scriptsize}
}

% --- coordinates (chosen to match your Sage layout qualitatively) ---
% Layer 1
\node[pt] (x1_2) at (0,  2.2) {};
\node[pt] (x1_1) at (0,  1.2) {};

% Layer 2 (top to bottom: x2_4, x2_3, x2_2, x2_1)
\node[pt] (x2_4) at (3,  2.7) {};
\node[pt] (x2_3) at (3,  2.2) {};
\node[pt] (x2_2) at (3,  1.2) {};
\node[pt] (x2_1) at (3,  0.7) {};

% Layer 3 (top to bottom: x3_7, x3_6, x3_5, x3_4, x3_3, x3_2, x3_1)
\node[pt] (x3_7) at (6,  3.2) {};
\node[pt] (x3_6) at (6,  2.7) {};
\node[pt] (x3_5) at (6,  2.2) {};
\node[pt] (x3_4) at (6,  1.7) {};
\node[pt] (x3_3) at (6,  1.2) {};
\node[pt] (x3_2) at (6,  0.7) {};
\node[pt] (x3_1) at (6,  0.2) {};

% --- edges layer 1 -> 2 ---
\draw[ed] (x1_2) -- (x2_4);
\draw[ed] (x1_2) -- (x2_3);
\draw[ed] (x1_1) -- (x2_2);
\draw[ed] (x1_1) -- (x2_1);

% --- edges layer 2 -> 3 ---
\draw[ed] (x2_4) -- (x3_7);
\draw[ed] (x2_4) -- (x3_6);

\draw[ed] (x2_3) -- (x3_5);
\draw[ed] (x2_3) -- (x3_4);

\draw[ed] (x2_2) -- (x3_3);

\draw[ed] (x2_1) -- (x3_2);
\draw[ed] (x2_1) -- (x3_1);

% --- labels (xI_a -> x^{(I)}_a), placed to the right like your picture ---
\node[lab, anchor=west] at ($(x1_2)+(0.25,-0.15)$) {$x^{(1)}_{2}$};
\node[lab, anchor=west] at ($(x1_1)+(0.25,-0.15)$) {$x^{(1)}_{1}$};

\node[lab, anchor=west] at ($(x2_4)+(0.25,-0.15)$) {$x^{(2)}_{4}$};
\node[lab, anchor=west] at ($(x2_3)+(0.25,-0.15)$) {$x^{(2)}_{3}$};
\node[lab, anchor=west] at ($(x2_2)+(0.25,-0.15)$) {$x^{(2)}_{2}$};
\node[lab, anchor=west] at ($(x2_1)+(0.25,-0.15)$) {$x^{(2)}_{1}$};

\node[lab, anchor=west] at ($(x3_7)+(0.25,-0.15)$) {$x^{(3)}_{7}$};
\node[lab, anchor=west] at ($(x3_6)+(0.25,-0.15)$) {$x^{(3)}_{6}$};
\node[lab, anchor=west] at ($(x3_5)+(0.25,-0.15)$) {$x^{(3)}_{5}$};
\node[lab, anchor=west] at ($(x3_4)+(0.25,-0.15)$) {$x^{(3)}_{4}$};
\node[lab, anchor=west] at ($(x3_3)+(0.25,-0.15)$) {$x^{(3)}_{3}$};
\node[lab, anchor=west] at ($(x3_2)+(0.25,-0.15)$) {$x^{(3)}_{2}$};
\node[lab, anchor=west] at ($(x3_1)+(0.25,-0.15)$) {$x^{(3)}_{1}$};

\end{tikzpicture}
\end{center}
The list of monomials is 
\[
1,\;\cC_{1},\;\cC_{3},\;\cC_{1}\cC_{3},\;\cC_{2},\;\cC_{1}\cC_{2},\;\cC_{2}\cC_{3}.
\]
The matrix is 
\[
M=\begin{pmatrix}
1 & x^{(1)}_{1} & x^{(3)}_{1} & x^{(1)}_{1}x^{(3)}_{1} & x^{(2)}_{1} & x^{(1)}_{1}x^{(2)}_{1} & x^{(2)}_{1}x^{(3)}_{1}\\
1 & x^{(1)}_{1} & x^{(3)}_{2} & x^{(1)}_{1}x^{(3)}_{2} & x^{(2)}_{1} & x^{(1)}_{1}x^{(2)}_{1} & x^{(2)}_{1}x^{(3)}_{2}\\
1 & x^{(1)}_{1} & x^{(3)}_{3} & x^{(1)}_{1}x^{(3)}_{3} & x^{(2)}_{2} & x^{(1)}_{1}x^{(2)}_{2} & x^{(2)}_{2}x^{(3)}_{3}\\
1 & x^{(1)}_{2} & x^{(3)}_{4} & x^{(1)}_{2}x^{(3)}_{4} & x^{(2)}_{3} & x^{(1)}_{2}x^{(2)}_{3} & x^{(2)}_{3}x^{(3)}_{4}\\
1 & x^{(1)}_{2} & x^{(3)}_{5} & x^{(1)}_{2}x^{(3)}_{5} & x^{(2)}_{3} & x^{(1)}_{2}x^{(2)}_{3} & x^{(2)}_{3}x^{(3)}_{5}\\
1 & x^{(1)}_{2} & x^{(3)}_{6} & x^{(1)}_{2}x^{(3)}_{6} & x^{(2)}_{4} & x^{(1)}_{2}x^{(2)}_{4} & x^{(2)}_{4}x^{(3)}_{6}\\
1 & x^{(1)}_{2} & x^{(3)}_{7} & x^{(1)}_{2}x^{(3)}_{7} & x^{(2)}_{4} & x^{(1)}_{2}x^{(2)}_{4} & x^{(2)}_{4}x^{(3)}_{7}
\end{pmatrix}.
\]
The determinant is 
\bea\label{DetMexamp2} 
\det M
=\bigl(x^{(1)}_{1}-x^{(1)}_{2}\bigr)^{3}
\bigl(x^{(2)}_{1}-x^{(2)}_{2}\bigr)
\bigl(x^{(2)}_{3}-x^{(2)}_{4}\bigr)^{2}
\bigl(x^{(3)}_{1}-x^{(3)}_{2}\bigr)
\bigl(x^{(3)}_{4}-x^{(3)}_{5}\bigr)
\bigl(x^{(3)}_{6}-x^{(3)}_{7}\bigr).
\eea
It is a generalized Vandermonde with factors which ensure that the determinant is not vanishing when the eigenvalues of nodes in the first layer are distinct, and the eigenvalues of nodes in subsequent layers sharing the same parent are distinct. These inequalities precisely reflect the fact that  the edges of the  degeneracy graph from one layer to the next correspond to the resolution  the projectors of $ \cA_i$ into those of $ \cA_{i+1}$. This is expressed in equation \eqref{Expitoipl1}  by the apperance of the blocks $B^{(i+1)}_a$ which encode the connections in the graph.

As in the first example, the determinant is exactly equal, up to a sign, to the product over all pairs  of nodes in the first layer, along with all pairs in subsequent layers sharing the same parent : for each pair we have the difference of eigenvalues raised to a positive power. A neater way to state this is to add a zeroth layer corresponding to the one dimensional algebra $ \cA_0$ with a single node, and connect the single node to every node in the first layer where the nodes correspond to the projectors in $\cA_1$. The degeneracy graph then becomes a rooted layered graph.  Then the neater statement is that the determinant is, up a sign, the product over all pairs of nodes in layer one and higher, which share the same parent. For each pair the determinant has a factor of the corresponding difference of eigenvalues raised to a positive power. This is Determinant Conjecture 1 \eqref{FactoredDet}.  The code for systematic checks of this conjecture is in the file DetConjecture1.ipynb. Using input $ D_1 , D_2 , \cdots , D_L$, the code scans through the different possible partitions and compositions, and reports  on the validity of Conjecture 1. 

The usage of the conjecture checker DetConjecture1.ipynb, for a specific choice of data for dimensions $D_i$ and partition/compositions is illustrated by the following example. 
\begin{verbatim} 
D = [2, 4, 7, 12 ]
p1 = [2, 2]
comps = {2: [2, 1, 2, 2] , 3: [2,3,1,2,2,1,1  ] } 
res = test_refined_conjecture(D, p1, comps)
\end{verbatim} 
The output reports the positive outcome of the test as 
\begin{verbatim} 
missing sibling pairs = 0
\end{verbatim} 
The usage of DetConjecture1.ipynb, for a sweep over the partitions and compositions with fixed dimensions $D_i$ is illustrated by: 
\begin{verbatim} 
sweep_refined_conjecture_fixed_D([2,5,9])
\end{verbatim} 
The output illustrating successful check is: 
\begin{verbatim} 
=== Sweep summary (refined conjecture) ===
D = [2, 5, 9]
tested = 140 (out of total 140)
failures = 0
{'D': [2, 5, 9], 'tested': 140, 'total': 140, 'failures': []}
\end{verbatim} 

We can see in the two examples above, that the differences of eigenvalues associated with nodes at level $L$ appear with power equal to $1$. Note also the factor $(x^{(2)}_{3}-x^{(2)}_{4})^{2}$ in \eqref{DetMexamp2}. If the second graph is truncated by dropping all the nodes at layer 1, and all the nodes at layer 2 except the two with eigenvalue labels  $  x^{(2)}_{3} , x^{(2)}_{4} $, and we further only keep the descendants of these two nodes at layer $3$, we get the graph in example 1, up to renaming of the eigenvalue labels. Importantly, the exponent of the difference in eigenvalues for the truncated graph is the same as before truncation : that it, the exponent of $2$ also occurs for $   ( x^{(1)}_{1} - x^{(1)}_{2} )  $ in \eqref{DetExamp1}. The exponent $2$ is also equal to the number of monomials in the truncated graph having a non-zero power of 
$ \cC_1$. The positivity of this number is easy to see by applying \eqref{mainconj} to the truncated graph.

The second determinant conjecture \eqref{DetConj2} states the exponent for any eigenvalue difference at layers 
$ \{ 1 , \cdots , L -1 \}$ is thus preserved under the graph truncation operation we just described, and is equal to the number of monomials containing $ \cC_1$ in the monomial basis for 
the truncated graph. The positivity of the latter number is evident from \eqref{mainconj}. 
The file DetConjecture2.ipynb gives systematic tests of this conjecture for specified $ \{ D_1 , D_2 , \cdots , D_L \}$.

The usage of DetConjecture2.ipynb for fixed $D_i$ and fixed choice of partitions/compositions is  illustrated by 
\begin{verbatim} 
D = [2, 4, 7, 12 ]
p1 = [2, 2]
comps = {2: [2, 1, 2, 2] , 3: [2,3,1,2,2,1,1  ] } 
test_conjecture2(D, p1, comps)
\end{verbatim} 
The successful outcome of the test is illustrated by 
\begin{verbatim} 
=== Conjecture 2 test ===
D=[2, 4, 7, 12], p1=[2, 2], comps={2: [2, 1, 2, 2], 3: [2, 3, 1, 2, 2, 1, 1]}
mismatches = 0
\end{verbatim} 
The usage for a sweep over the partitions and compositions at fixed $D_i$ is illustrated by 
\begin{verbatim}
summary = sweep_conjecture2_fixed_D([3,5,8])
\end{verbatim}  
with positive outcome evidenced by
\begin{verbatim}
=== Sweep summary (Conjecture 2) ===
D = [3, 5, 8]
tested = 70 (out of total 70)
failures = 0
\end{verbatim}

\vskip.5cm

\end{appendix} 

%%%%%%%%%%%%%%%%%%%%%%%%%%%%%%%%%%%%%%%%%%%%%%%%%%%%%%%%%%%%%%%%%%%%%%%%%%%%%%

\end{document}